%% file: npi_arxiv.tex
\providecommand{\blind}{0}
\providecommand{\draftnotes}{0}

\documentclass[11pt,letterpaper]{article}
\usepackage[margin=1in]{geometry}
\usepackage{iftex}
\usepackage[T1]{fontenc}
\ifPDFTeX
  \usepackage[utf8]{inputenc}
\fi
\usepackage{lmodern,microtype}
\usepackage{amsmath,amssymb,amsthm}
\usepackage{mathrsfs}
\usepackage{needspace}
\usepackage{graphicx,helvet,xcolor}
\usepackage{algorithm}
\usepackage{algpseudocode}
\usepackage[authoryear,round]{natbib}
\usepackage{tikz}
\usetikzlibrary{arrows.meta,decorations.pathreplacing,calc,positioning,patterns,plotmarks}
\usepackage{pgfplots}
\pgfplotsset{compat=1.18}
\usepackage{caption}
\usepackage[colorlinks=true,linkcolor=blue,citecolor=blue,urlcolor=blue]{hyperref}
\IfFileExists{xurl.sty}{\usepackage{xurl}}{}
\providecommand{\harvardurl}[1]{\url{#1}}
\renewcommand{\harvardurl}[1]{\url{#1}}

\newcommand{\spacingset}[1]{\renewcommand{\baselinestretch}{#1}\small\normalsize}
\newcommand{\Prob}{\mathbb{P}}
\newcommand{\ind}{\mathbf{1}}
\newcommand{\E}{\mathbb{E}}

\numberwithin{equation}{section}
\newtheorem{proposition}{Proposition}[section]
\newtheorem{theorem}[proposition]{Theorem}
\newtheorem{corollary}[proposition]{Corollary}
\newtheorem{lemma}[proposition]{Lemma}
\theoremstyle{remark}
\newtheorem{remark}[proposition]{Remark}
\theoremstyle{plain}
\algrenewcommand{\algorithmicrequire}{\textbf{Input:}}
\algrenewcommand{\algorithmicensure}{\textbf{Output:}}
\newif\ifblinded
\ifnum\blind=1\relax\blindedtrue\else\blindedfalse\fi
\newcommand{\draftnote}[1]{%
  \ifblinded\else\ifnum\draftnotes=1\relax{\color{black}#1}\fi\fi}

\newcommand{\papertitle}{Neyman--Pearson Classification with Indecisions:
Simultaneous Control of Selective Type I and Type II Errors}
\title{\papertitle}
\date{}
\ifblinded
  \author{}
  \hypersetup{pdftitle={\papertitle},pdfauthor={},pdfsubject={},pdfkeywords={}}
\else
  \author{Bradley Rava\thanks{University of Sydney Business School.}
    \and Bowen Gang\thanks{University of Alberta, Department of Mathematical and Statistical Sciences.}}
  \hypersetup{pdftitle={\papertitle},pdfauthor={Bradley Rava and Bowen Gang}}
\fi

\begin{document}
\pagenumbering{arabic}
\spacingset{1}
\maketitle

\begin{abstract}
Binary classification often requires controlling both type I and type II error rates at prespecified levels. When the class distributions overlap substantially, the two error limits may be incompatible if every observation must be classified. We study Neyman--Pearson classification with indecisions (NPI), which permits uncertain observations to remain unclassified and assesses selective type I and type II errors among definitive classifications within each class. For an independently trained score, we develop a model-free calibration procedure that combines all candidate threshold pairs passing calibration into a rule whose decision region contains that of every passing candidate. We couple the entire threshold search to a common martingale within each class and compute its finite-horizon crossing probability exactly, avoiding a multiplicity adjustment for the number of candidate pairs. The resulting classifier satisfies both population selective error limits simultaneously with user-specified high probability in finite samples, without assumptions on the accuracy of the fitted score. We bound the additional indecision due to calibration and evaluate the procedure in simulations and applications to recidivism and credit-default prediction.

\end{abstract}

\medskip

\section{Introduction}
\label{sec:introduction}
Artificial intelligence systems are used to support decisions in medicine, finance, and criminal justice, where incorrect classifications can have serious consequences for the individuals involved. Examples include identifying patients with a disease, detecting fraudulent transactions, and assessing recidivism risk. Although advances in machine learning have improved predictive performance, high accuracy alone does not ensure that a classifier's errors remain below levels acceptable to the user. Consequential decisions therefore call for explicit error guarantees at levels chosen by the user.

A binary classifier can make an error by assigning an observation from class 1 to class 2, or an observation from class 2 to class 1. In Neyman--Pearson (NP) classification, the corresponding class-conditional misclassification probabilities are termed the type I and type II error rates, respectively \citep{tong2018neyman}. These errors can have different consequences. In disease screening, classifying a healthy individual as diseased may lead to unnecessary follow-up, whereas classifying a diseased individual as healthy may delay diagnosis. Both class-conditional error rates may therefore need to be controlled, potentially at different levels. The NP paradigm addresses this asymmetry by minimizing type II error subject to a prescribed limit on type I error \citep{scott2005neyman}. However, the minimum achievable type II error under this constraint may still be unacceptably large. When the class distributions overlap substantially, even an oracle knowing both distributions may be unable to meet the desired limits on both errors if every observation must be classified.

One way to address this difficulty is to allow an indecision for observations that cannot be classified with sufficient confidence \citep{herbei2006classification, SunWei11, rava2026burden}. These observations can be referred for further assessment, while definitive decisions are made on the remaining observations. This selective approach shifts attention to the reliability of the decisions actually made. A common criterion in this framework is the false selection rate (FSR), defined as the expected fraction of incorrect classifications among observations receiving a definitive decision. Class-specific FSR criteria assess this fraction separately among observations assigned to each class \citep{gang2024lass,sun2025unified}.

However, even class-specific FSR control does not generally provide class-conditional error control. Class-specific FSR criteria assess errors among observations \emph{assigned to} a class, rather than among classified observations \emph{belonging to} that class. In disease screening, FSR control for positive predictions limits the expected fraction of healthy individuals among those labeled positive. This fraction can be small even when a large proportion of the healthy individuals receiving a definitive decision are labeled positive. A small FSR among negative predictions does not rule out this possibility. Hence, reliable predictions for both labels do not necessarily imply low misclassification probabilities within both classes.

This article studies Neyman--Pearson classification with indecisions (NPI) to address this issue. We consider the \emph{selective type I and type II error rates}, which measure the probabilities of misclassification among observations from class 1 and class 2, respectively, conditional on receiving a definitive decision. The two error limits can be chosen separately to reflect the different consequences of misclassification for the two classes. Our goal is to minimize the proportion of indecisions subject to both error constraints.

\begin{figure}[t!]
    \centering
    \input{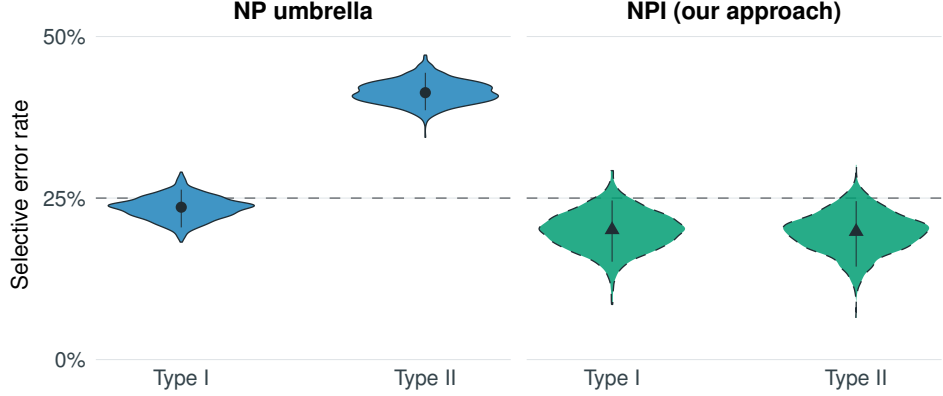}
    \caption{\label{compas_intro.plot}\small Comparison of NP and NPI classification on COMPAS. Left: NP umbrella, which classifies every observation. Right: NPI, which permits indecisions. Violins show held-out errors across all 1,000 partitions from Section~\ref{subsec:compas-application}. Symbols mark medians, vertical intervals span the 5th--95th percentiles, and green ticks at zero are labeled with the percentage of partitions with zero error. Dashed lines mark the $25\%$ target, imposed on type I error for NP and jointly on both selective errors for NPI.
    }
\end{figure}

Figure~\ref{compas_intro.plot} contrasts conventional NP classification with our NPI procedure using the COMPAS recidivism data \citep{Angetal16-COMPAS1,Dieetal16-COMPAS2}. The left panel shows the NP umbrella classifier of \citet{tong2018neyman}, which assigns a class label to every individual. It controls the population type I error at the $25\%$ target with high probability, but places no corresponding bound on type II error. In the displayed results, the held-out type I errors are centered below the target, whereas the type II errors lie well above it. The right panel shows our NPI procedure, which permits indecisions. The procedure provides a joint high-probability guarantee that both population selective errors are at most $25\%$. Both plotted selective error distributions are centered below this common target.

Obtaining this joint population guarantee from finite data requires valid calibration of the classification thresholds. In practice, a fitted classifier supplies a score, but the population selective errors associated with candidate thresholds are unknown. Choosing thresholds that satisfy both population constraints therefore requires more than checking whether their empirical errors are below the desired levels. We develop a data-driven NPI algorithm that tests candidate threshold pairs using a labeled calibration sample independent of score training. The decision regions of all pairs passing both calibration tests can be combined into one rule that also passes and classifies at least as many observations as any of these candidates.

This paper makes several contributions. First, we develop a model-free calibration algorithm that accommodates scores from any independently trained classifier, including complex black-box models. Second, we establish a joint finite-sample guarantee: both population selective error limits hold simultaneously with probability at least $1-\delta$, for a user-specified $\delta$. A common martingale construction bounds the probability that any threshold pair violating a class's error limit passes its check. The martingale's exact crossing probability determines the critical values. Third, the selected rule maximizes the decision probability within each class among candidates passing calibration, without estimating class proportions. The error guarantee remains valid under label shift, where class proportions change while the feature distributions within classes remain fixed. Finally, simulations and real-data analyses illustrate the gains in coverage, including substantial improvements over both Learn then Test implementations \citep{angelopoulos2025learn} in the nonlinear setting and the applications.

\subsection{Related work}
\label{subsec:related-work}

NP classification develops decision rules subject to explicit class-conditional error constraints \citep{scott2005neyman,rigollet2011neyman}. The NP umbrella algorithm provides high-probability type I error control for a broad range of scoring classifiers \citep{tong2018neyman}, and \citet{tian2025neyman} extend the NP formulation to multiple error constraints in multiclass problems. Most closely related, \citet{ndaoud2026ask} study the same selective type I and type II error criteria and the objective of minimizing indecisions under both error constraints. They characterize oracle rules and propose empirical threshold calibration. Our work develops a calibration procedure with joint finite-sample guarantees for both population errors.

Allowing indecisions also connects our work to classification with a reject option \citep{chow1970optimum,herbei2006classification}. Selective classification studies the trade-off between the fraction of predictions retained (coverage) and the overall error among them \citep{elyaniv2010foundations,geifman2017selective}. FSR-based methods control the expected proportion of errors among selected predictions, with guarantees for individual predicted classes and, in some settings, protected groups \citep{gang2024lass,sun2025unified,rava2026burden}. Related work uses conformal $p$-values for predictive selection with false discovery rate control \citep{jin2023selection, marandon2024adaptive, gazin2025selecting}. Class-specific FSR criteria group observations by assigned labels; our selective errors condition on class membership among definitive decisions.

Set-valued classification offers another way to represent uncertainty by returning a set of plausible labels \citep{lei2014classification,sadinle2019least}. These methods seek informative sets subject to prediction-set coverage requirements, which may be imposed separately within each class. Conformal prediction also provides prediction-set coverage guarantees, including class-conditional coverage \citep{shafer2008tutorial,vovk2013conditional,romano2020classification}. Such guarantees concern whether the prediction set contains the correct label. They do not by themselves control the selective errors obtained by retaining only singleton predictions, because the latter additionally condition on receiving a definitive decision.

Finite-sample calibration translates population error requirements into data-driven decision rules. \citet{bates2021riskcontrolling} construct risk-controlling prediction sets for general loss functions. The Learn then Test (LTT) framework of \citet{angelopoulos2025learn} allows simultaneous control of multiple risks and data-driven selection among candidates passing its tests, including optimization of an empirical performance criterion using the calibration data. Its calibration step uses a multiple-testing procedure to account for selection among candidate rules. Single-start fixed-sequence testing avoids a candidate-count adjustment, but requires an ordering fixed independently of calibration and stops at the first nonrejection. Recent work also provides finite-grid guarantees for multiple selective error metrics \citep{jemelen2026beyond} and for selective risk, acceptance probability, and utility \citep{yu2026jointcertificate}, with adjustments for the grid size. Uniform confidence bands for the class-conditional score distributions, such as those supplied by the Dvoretzky--Kiefer--Wolfowitz inequality \citep{massart1990tight}, offer another route to simultaneous error bounds. A key innovation of our approach is a common martingale construction that directly targets the selective error constraints within each class and provides simultaneous validity for all passing threshold pairs, without a prescribed testing order or a multiplicity adjustment for the number of pairs searched. We calibrate the critical values using the exact finite-horizon crossing probability of this common process. We then combine the decision regions of passing pairs to maximize coverage while preserving joint finite-sample error control.

\subsection{Paper organization}
\label{subsec:paper-organization}

Section~\ref{sec:problem} formulates the problem and derives an oracle rule. Section~\ref{sec:methodology} presents NPI, its error guarantee, and the cost of calibration. Sections~\ref{sec:simulations} and~\ref{sec:real_data} report simulations and the COMPAS application, followed by a discussion in Section~\ref{sec:discussion}. The appendix contains proofs, additional theory, implementation details, and the credit-card analysis.

\section{Problem formulation}
\label{sec:problem}

Suppose we observe a labeled data set $\mathcal D^{\mathrm{obs}}=\{(X_j,Y_j):j=1,\ldots,n\}$, where $X_j\in\mathcal X$ is a feature vector and $Y_j\in\mathcal C=\{1,2\}$ is the corresponding class label. Our goal is to use these data to classify new observations with explicit control of both type I and type II errors. For a generic observation $(X,Y)$ from the target population, let $P_i$ denote the conditional distribution of $X$ given $Y=i$, and let $\pi_i=\Prob(Y=i)>0$, $i=1,2$. We assume that the observed and future feature vectors are independent conditional on their class labels, with common distribution $P_i$ within class $i$. Class proportions may differ between the observed data and the target population; $\pi_i$ denotes the target proportion.

\subsection{Background: type I and type II errors}
\label{subsec:ordinary-errors}

A standard classification rule $g:\mathcal X\to\mathcal C$ assigns every observation to one of the two classes. Taking class 1 as the null state, the type I and type II error rates are, respectively,
\begin{equation}
 R_i^{\mathrm{full}}(g)
=\Prob\{g(X)\ne Y\mid Y=i\},
\qquad i=1,2.
 \label{eq:ordinary-errors}
\end{equation}
These rates measure the probability of misclassification conditional on the actual class.

In applications where the two types of error have different consequences, it is often desirable to give one error priority. The NP paradigm minimizes type II error subject to a prespecified type I error constraint \citep{cannon2002learning,scott2005neyman,tong2013plugin}. For a user-specified level $\alpha_1\in(0,1)$, the objective is to
\begin{equation}
 \underset{g}{\text{minimize}}\; R_2^{\mathrm{full}}(g)
 \quad\text{subject to}\quad R_1^{\mathrm{full}}(g)\le\alpha_1.
 \label{eq:ordinary-np}
\end{equation}
The resulting type II error can be large, as Figure~\ref{compas_intro.plot} illustrates for COMPAS.

\subsection{A selective notion of classification error}
\label{subsec:errors}

To address this difficulty, we allow the rule to return an indecision. Consider a rule $h:\mathcal X\to\Lambda$ with action space $\Lambda=\{0,1,2\}$. The action $h(X)=i$ assigns an observation to class $i$, whereas $h(X)=0$ returns an \emph{indecision}, corresponding to the reject option in classification \citep{herbei2006classification,SunWei11}. We define the \emph{selective type I and type II error rates} by
\begin{equation}
	R_i(h)=\Prob\{h(X)\ne Y\mid Y=i,\ h(X)\ne0\},
	\qquad i=1,2.
	\label{eq:risks}
\end{equation}
Each rate measures the probability of an incorrect classification among observations that actually belong to class $i$ and receive a definitive decision. Equivalently,
\[
R_1(h)=\frac{P_1\{h(X)=2\}}{P_1\{h(X)\ne0\}},
\qquad
R_2(h)=\frac{P_2\{h(X)=1\}}{P_2\{h(X)\ne0\}},
\]
with the convention $0/0=0$. These criteria combine the class-conditional perspective of the NP paradigm with conditioning on selection. Controlling only the probabilities in the numerators above can be misleading: a rule may make few errors simply because it rarely makes a decision. The denominators restrict the assessment to the observations classified within each actual class. When there are no indecisions, $R_i(h)=R_i^{\mathrm{full}}(h)$, so the selective error rates reduce to the usual type I and type II error rates.

Conditioning on the actual class distinguishes the selective type I and type II error rates from the pooled selective risk, $\Prob\{h(X)\ne Y\mid h(X)\ne0\}$, which pools all definitive classifications \citep{elyaniv2010foundations,geifman2017selective}. Such a pooled criterion can be small even when the selective error $R_i(h)$ is large for a class that contributes few decisions. The FSR evaluates the expected fraction of errors in a finite set of selected observations, with the false discovery rate of \citet{benjamini1995controlling} as its multiple-testing special case. Class-specific FSR criteria further distinguish errors among observations \emph{assigned to} each class \citep{gang2024lass,rava2026burden}. For example, when $\Prob\{h(X)=2\}>0$, the population counterpart of the FSR for selecting class 2 is $\Prob\{Y=1\mid h(X)=2\}$. In contrast, $R_1(h)$ assesses errors among observations \emph{from class 1} that receive a definitive decision. Consequently, pooled or predicted-class error control does not generally ensure that both selective type I and type II errors are small.

Allowing indecisions can reduce classification errors, but leaves some observations unclassified. To quantify this cost, define the population indecision rate by
\begin{equation}
 \Gamma(h)=\Prob\{h(X)=0\}
 =\sum_{i=1}^2\pi_iP_i\{h(X)=0\}.
 \label{eq:indecision-rate}
\end{equation}
The coverage is $C(h)=1-\Gamma(h)=\sum_{i=1}^2\pi_iq_i(h)$, where $q_i(h)=P_i\{h(X)\ne0\}$ is the decision probability within class $i$. For prespecified levels $\alpha_1,\alpha_2\in(0,1)$, our goal is to
\begin{equation}
 \underset{h}{\text{minimize}}\;\Gamma(h)
 \quad\text{subject to}\quad
 R_1(h)\le\alpha_1,\qquad R_2(h)\le\alpha_2.
 \label{eq:population}
\end{equation}
The optimization is over all measurable classification rules, with randomization allowed. For a randomized rule, the probabilities in \eqref{eq:risks} and \eqref{eq:indecision-rate} also average over the decision randomization.

\subsection{The oracle rule}
\label{subsec:oracle}

\begingroup\color{black}
With the population distributions known, an oracle rule $h^{\mathrm{OR}}$ solves \eqref{eq:population}. We show that such a rule can be obtained by thresholding the likelihood ratio. Let $f_i$ denote the density of $X$ given $Y=i$, $i=1,2$, and define
\[
 L(x)=\frac{f_1(x)}{f_2(x)}.
\]
Large values of $L(x)$ favor class 1, whereas small values favor class 2. Consider rules of the form
\begin{equation}
 h^L_{\ell,u}(x)
 =2\cdot\ind\{L(x)<\ell\}+\ind\{L(x)>u\},
 \label{eq:oracle}
\end{equation}
where $\ind\{\cdot\}$ is the indicator function and $0\le\ell\le u\le+\infty$. Equation \eqref{eq:oracle} specifies the decision when $L(x)\notin\{\ell,u\}$. At $L=\ell<u$, randomization between class 2 and indecision is allowed; at $L=u>\ell$, randomization between class 1 and indecision is allowed. When $\ell=u$, all three actions are allowed at the common threshold.

We choose the thresholds to minimize indecision while satisfying both error constraints:
\begin{equation}
 (\ell^{\mathrm{OR}},u^{\mathrm{OR}})
 \in\underset{\substack{0\le\ell\le u\le+\infty\\
 R_i(h^L_{\ell,u})\le\alpha_i,\ i=1,2}}
 {\operatorname{arg\,min}}\;\Gamma(h^L_{\ell,u}).
 \label{eq:oracle-cutoffs}
\end{equation}
The minimization is also over the randomization probabilities at $L=\ell$ and $L=u$. The following theorem shows that this threshold search attains the optimum over all classification rules in \eqref{eq:population}.

\Needspace{9\baselineskip}
\begin{theorem}
\label{thm:oracle}
Let $\mathcal H_{\boldsymbol\alpha}$ be the collection of all measurable, possibly randomized classifiers satisfying $R_i(h)\le\alpha_i$, $i=1,2$. For $\pi_i>0$ and $\alpha_i\in(0,1)$, $i=1,2$, problem \eqref{eq:population} admits an optimal rule of the form \eqref{eq:oracle}, with boundary randomization if needed. Consequently, \eqref{eq:oracle-cutoffs} has a solution, and any resulting rule $h^{\mathrm{OR}}$ satisfies
\[
 \Gamma(h^{\mathrm{OR}})
 =\min_{h\in\mathcal H_{\boldsymbol\alpha}}\Gamma(h).
\]
If some rule without indecisions satisfies both constraints, the oracle thresholds can be chosen to coincide and $\Gamma(h^{\mathrm{OR}})=0$.
\end{theorem}

A related characterization in Theorem 2 of \citet{ndaoud2026ask} minimizes selective type II error at fixed indecision mass and selective type I level. Theorem~\ref{thm:oracle} establishes existence and the threshold representation directly for \eqref{eq:population}, including randomized boundary decisions.

Any strictly increasing transformation of $L$, including the posterior probability $\Prob(Y=1\mid X=x)$, gives the same decisions after adjusting the thresholds. This motivates learning a score that approximates this ordering and calibrating its thresholds to meet both error requirements.
\endgroup

\subsection{The data-driven objective}
\label{subsec:objective}

The oracle rule provides a benchmark for \eqref{eq:population}, but depends on the unknown class-conditional distributions $P_1,P_2$. In practice, both the score and its thresholds must be learned from data. The resulting population selective errors are random, and controlling their empirical counterparts does not ensure that both population constraints are satisfied.

Following the finite-sample NP framework \citep{rigollet2011neyman,tong2018neyman,yao2023asymmetric}, we seek a data-driven rule $\widehat h$ that meets both prescribed error limits with high probability. Specifically, for a user-specified failure tolerance $\delta\in(0,1)$, we require
\begin{equation}
 \Prob\!\left\{
 R_1(\widehat h)>\alpha_1
 \ \text{or}\
 R_2(\widehat h)>\alpha_2
 \right\}\le\delta,
 \label{eq:joint-control}
\end{equation}
where the probability is taken over the data and any randomness used to construct the rule. Our objective is to establish this joint guarantee in finite samples while keeping the population indecision rate $\Gamma(\widehat h)$ as small as possible.

\section{Methodology}
\label{sec:methodology}

The oracle analysis suggests assigning class labels to the two tails of a classification score and allowing indecisions in between. To implement this approach, we split $\mathcal D^{\mathrm{obs}}$ into a training sample $\mathcal D^{\mathrm{train}}$ and an independent calibration sample $\mathcal D^{\mathrm{cal}}$. Let $s:\mathcal X\to\mathbb R$ be a score fitted and tuned using only $\mathcal D^{\mathrm{train}}$, with larger values favoring class 1. Ideally, $s$ is a strictly increasing transformation of the likelihood ratio $L$. Our error guarantee requires no assumption on the accuracy of $s$ or its ordering relative to $L$. We use $\mathcal D^{\mathrm{cal}}$ to choose the two thresholds and achieve the joint error guarantee in \eqref{eq:joint-control}.

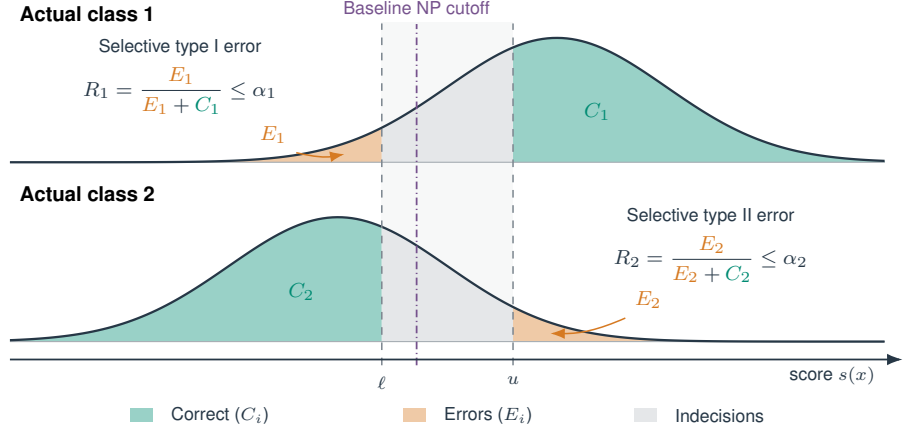
\begin{figure}[t]
 \centering
 \resizebox{0.80\textwidth}{!}{%
  \input{arxiv_figures/npi_threshold_cartoon.tex}%
 }
 \caption{\small NPI assigns the lower and upper score tails to classes 2 and 1, respectively, and returns indecisions between them. Within actual class $i$, $E_i$ and $C_i$ are the probabilities of an incorrect and a correct decision; the selective error rate is $E_i/(E_i+C_i)$, excluding indecisions.
 }
 \label{fig:npi_schematic}
\end{figure}

\subsection{Calibration for selective error control}
\label{subsec:method-scores}
\label{subsec:method-selection}

Throughout this section, we condition on the training information (the training data and any independent training randomization), so that $s$ is fixed. For a candidate threshold pair $-\infty\le\ell\le u\le+\infty$, define
\begin{equation}
 h_{\ell,u}(x)
 =2\cdot \ind\{s(x)\le\ell\}+\ind\{s(x)>u\}.
 \label{eq:candidate-rule}
\end{equation}
The rule assigns class 2 to the lower score tail, class 1 to the upper tail, and an indecision to the region $(\ell,u]$; see Figure~\ref{fig:npi_schematic}.

Conditional on the calibration labels, write $\mathcal D^{\mathrm{cal}}=\{(X_{ij},i):j=1,\ldots,n_i,\ i=1,2\}$, where the $X_{ij}\sim P_i$ are independent and $n_i\ge1$ is the number of observations in class $i$. For each pair $(\ell,u)$, let
\begin{equation}
 \begin{aligned}
 m_i(\ell,u)&=\sum_{j=1}^{n_i}
       \ind\{h_{\ell,u}(X_{ij})\ne0\},\\
 e_i(\ell,u)&=\sum_{j=1}^{n_i}
       \ind\{h_{\ell,u}(X_{ij})=3-i\},
       \qquad i=1,2.
 \end{aligned}
 \label{eq:calibration-counts}
\end{equation}
denote the numbers of decisions and errors within class $i$. The number of correct decisions is $m_i(\ell,u)-e_i(\ell,u)$.

We assess a candidate rule by testing whether each selective error is below its target. Consider first a threshold pair fixed before calibration and the hypotheses
\begin{equation}
 H_{0i}:R_i(h_{\ell,u})\ge\alpha_i
 \qquad\text{versus}\qquad
 H_{1i}:R_i(h_{\ell,u})<\alpha_i,\quad i=1,2.
 \label{eq:calibration-hypotheses}
\end{equation}
Rejecting both null hypotheses provides evidence that the rule meets both error requirements. Conditional on $m_i(\ell,u)=m$, the error count $e_i(\ell,u)$ follows a $\operatorname{Bin}\{m,R_i(h_{\ell,u})\}$ distribution.

For $e$ errors and $k$ correct decisions, we compare an error rate $r<\alpha$ with the null boundary $\alpha$ through the binomial likelihood ratio of $r$ to $\alpha$. Averaging these likelihood ratios over $r\in(0,\alpha)$ gives the calibration statistic
\begin{equation}
 B_\alpha(e,k)
 =\frac1\alpha\int_0^\alpha
   \left(\frac r\alpha\right)^e
   \left(\frac{1-r}{1-\alpha}\right)^k\,dr.
 \label{eq:mixture-evidence}
\end{equation}
Larger values of $B_\alpha(e,k)$ provide stronger evidence for an error rate below $\alpha$: adding a correct decision increases the statistic, whereas adding an error decreases it.
\begingroup
To calibrate this statistic, let $\xi_1,\ldots,\xi_n$ be independent $\operatorname{Bernoulli}(\alpha)$ variables and write $W_m=\sum_{j=1}^m\xi_j$, with $W_0=0$. The reference process $T_m=B_\alpha(W_m,m-W_m)$ is a nonnegative martingale with $T_0=1$ \citep{shafer2011test}. Ville's inequality \citep{ville1939etude}, valid for any nonnegative martingale starting at one, gives, for $c\ge1$,
\begin{equation}
 \psi_{n,\alpha}(c)
 =\Prob\!\left\{\max_{1\le m\le n}T_m>c\right\}\le\frac1c.
 \label{eq:reference-calibration-tail}
\end{equation}
However, this bound uses only the martingale property. Here $T_m$ is a function of a Bernoulli random walk, allowing us to calculate $\psi_{n,\alpha}(c)$ exactly and obtain smaller critical values. For each decision count $m$, define the largest error count whose statistic exceeds $c$
\begin{equation}
 b_{\alpha,c}(m)
 =\max\bigl\{e\in\{0,\ldots,m\}:
                  B_\alpha(e,m-e)>c\bigr\},
 \label{eq:exact-boundary}
\end{equation}
with value $-1$ when the set is empty. In particular, $b_{\alpha,c}(0)=-1$.

\begin{lemma}
\label{lem:fixed-pair-calibration}
Fix $n\ge1$, $\alpha\in(0,1)$, and $c\ge1$. Set $p_0(0)=1$, and take $p_m(w)=0$ outside $0\le w\le m$. The recursion
\begin{equation}
 \begin{aligned}
 p_m(w)
 &=\bigl\{\alpha p_{m-1}(w-1)+(1-\alpha)p_{m-1}(w)\bigr\}
      \ind\{w>b_{\alpha,c}(m)\},\\
 \psi_{n,\alpha}(c)&=1-\sum_{w=0}^n p_n(w)
 \end{aligned}
 \label{eq:crossing-recursion}
\end{equation}
computes the crossing probability exactly, and $\psi_{n,\alpha}(c)\le1/c$. For a threshold pair fixed before calibration, if $R_i(h_{\ell,u})\ge\alpha_i$, then
\[
 \Prob\{B_{\alpha_i}(e_i,m_i-e_i)>c\}
 \le\psi_{n_i,\alpha_i}(c),
\]
conditional on the training information and calibration labels.
\end{lemma}

Choose per-class failure allowances $\eta_1,\eta_2\in(0,1)$ such that $(1-\eta_1)(1-\eta_2)\ge1-\delta$. The product form relies on the independence of the two classwise calibration samples. Our default is $\eta_1=\eta_2=1-\sqrt{1-\delta}$. We use the critical value
\begin{equation}
 c_{n,\alpha}(\eta)
 =\min\{c\ge1:\psi_{n,\alpha}(c)\le\eta\},
 \qquad c_i=c_{n_i,\alpha_i}(\eta_i).
 \label{eq:exact-cutoff}
\end{equation}
The minimum exists because $\max_{m\le n}T_m$ takes finitely many values and the strict crossing probability is right-continuous in $c$. Moreover, $1\le c_{n,\alpha}(\eta)<1/\eta$, where $1/\eta$ is the critical value supplied by Ville's inequality. Thus every count pair $(e,k)$ satisfying $B_\alpha(e,k)\ge1/\eta$ also passes the exact check, and additional count pairs may pass. The recursion costs $O(n^2)$ operations for each candidate critical value, and the critical value is obtained by a one-dimensional search. It depends only on the sample size, error target, and failure allowance, so it can be computed before examining the scores.

A pair \emph{passes calibration} if
\begin{equation}
 B_{\alpha_i}\{e_i(\ell,u),\,m_i(\ell,u)-e_i(\ell,u)\}
 >c_i,\qquad i=1,2.
 \label{eq:table-check}
\end{equation}
Equivalently, $e_i(\ell,u)\le b_{\alpha_i,c_i}\{m_i(\ell,u)\}$ for both classes. Let $\mathcal A$ denote the set of passing pairs, from which we select the final classifier.
\endgroup

\subsection{The NPI procedure}
\label{subsec:method-implementation}
\label{subsec:method-guarantee}

We select thresholds by combining the decision regions of all pairs passing calibration. Let $\mathcal G$ contain the distinct calibration scores from both classes together with $-\infty$ and $+\infty$. Thresholds on this grid represent all calibration decision patterns. Write $\mathcal A_{\mathcal G}=\mathcal A\cap(\mathcal G\times\mathcal G)$. When this set is nonempty, take the largest lower threshold and smallest upper threshold among the passing pairs:
\begin{equation}
 \ell_{\max}=\max_{(\ell,u)\in\mathcal A_{\mathcal G}}\ell,
 \qquad
 u_{\min}=\min_{(\ell,u)\in\mathcal A_{\mathcal G}}u.
 \label{eq:extreme-thresholds}
\end{equation}
If $\ell_{\max}\le u_{\min}$, these thresholds give the union of their decision regions. If $\ell_{\max}>u_{\min}$, the two tails overlap and a common threshold at $u_{\min}$ classifies every observation. Accordingly, define
\begin{equation}
 (\widehat\ell,\widehat u)
 =\bigl(\min\{\ell_{\max},u_{\min}\},\,u_{\min}\bigr),
 \qquad \widehat h=h_{\widehat\ell,\widehat u},
 \label{eq:selected-rule}
\end{equation}
and return $\widehat h\equiv0$ if $\mathcal A_{\mathcal G}$ is empty.

\begin{proposition}
\label{prop:largest-region}
If $\mathcal A_{\mathcal G}$ is nonempty, the pair in \eqref{eq:selected-rule} passes calibration and satisfies
\begin{equation}
 \{x:h_{\ell,u}(x)\ne0\}
 \subseteq\{x:\widehat h(x)\ne0\},
 \qquad (\ell,u)\in\mathcal A_{\mathcal G}.
 \label{eq:largest-decision-region}
\end{equation}
Thus $\widehat h$ maximizes decision probability among these candidates under every distribution of the score. For the same score, calibration data, grid, targets, and allowances, its decision region also contains that of every grid rule satisfying Ville's checks $B_{\alpha_i}(e_i,m_i-e_i)\ge1/\eta_i$, $i=1,2$.
\end{proposition}

Calibration and threshold selection use only classwise counts, so neither requires estimating class proportions. Algorithm~\ref{alg:calibration} summarizes the NPI procedure.

\begin{algorithm}[t!]
\spacingset{1}
\caption{Neyman--Pearson classification with indecisions (NPI)}
\label{alg:calibration}
\begin{algorithmic}[1]
 \Require Training sample $\mathcal D^{\mathrm{train}}$;
 independent calibration sample $\mathcal D^{\mathrm{cal}}$;
 targets $\alpha_1,\alpha_2$; joint failure tolerance $\delta$.
 \State \textbf{Fit the score.} Train and tune $s$ using
 $\mathcal D^{\mathrm{train}}$ only.
 \State \textbf{Compute the critical values.}
 If either class is absent from $\mathcal D^{\mathrm{cal}}$,
 return $\widehat h\equiv0$.
 Use $\eta_1=\eta_2=1-\sqrt{1-\delta}$ and compute
 $c_i=c_{n_i,\alpha_i}(\eta_i)$ from \eqref{eq:exact-cutoff}.
 \State \textbf{Form the candidates.} Score the calibration data;
 let $\mathcal G$ contain their distinct scores and $-\infty,+\infty$.
 \State \textbf{Assess the errors.} For each $\ell,u\in\mathcal G$
 with $\ell\le u$, count the decisions $m_i$ and errors $e_i$
 within each actual class. The pair passes calibration if
 $B_{\alpha_i}(e_i,m_i-e_i)>c_i$ for both $i=1,2$.
 \State \textbf{Combine the decision regions.} If any pair passes,
 compute $\ell_{\max}$ and $u_{\min}$ from
 \eqref{eq:extreme-thresholds} and set
 $(\widehat\ell,\widehat u)=(\min\{\ell_{\max},u_{\min}\},u_{\min})$.
 Otherwise, return $\widehat h\equiv0$.
 \State \textbf{Classify new observations.} If a pair was selected,
 assign class 2 when $s(x)\le\widehat\ell$, class 1 when
 $s(x)>\widehat u$, and an indecision otherwise.
\end{algorithmic}
\end{algorithm}

A natural concern is whether selecting the largest decision region from many passing pairs invalidates error control. Lemma~\ref{lem:fixed-pair-calibration} establishes validity only for a prespecified pair. The following theorem shows that NPI preserves joint error control after this selection.

\begingroup
\begin{theorem}
\label{thm:simultaneous}
Under the calibration sampling scheme in Section~\ref{subsec:method-scores}, suppose the measurable score $s$, targets $\alpha_i\in(0,1)$, and allowances $\eta_i\in(0,1)$ with $(1-\eta_1)(1-\eta_2)\ge1-\delta>0$ are fixed before calibration. Conditionally on the training information and calibration labels, when both class counts are positive, for each class $i$ and every $c\ge1$,
\begin{equation}
 \Prob\!\left\{
 \begin{array}{l}
 \exists\,\ell\le u:\ R_i(h_{\ell,u})>\alpha_i,\\
 B_{\alpha_i}\{e_i(\ell,u),m_i(\ell,u)-e_i(\ell,u)\}>c
 \end{array}\right\}
 \le\psi_{n_i,\alpha_i}(c).
 \label{eq:sharp-classwise-control}
\end{equation}
Equality holds if the class-$i$ score distribution is continuous. Consequently, with conditional probability at least $(1-\eta_1)(1-\eta_2)$, every pair passing \eqref{eq:table-check} satisfies both population error limits. In particular, Algorithm~\ref{alg:calibration} returns a classifier satisfying
\[
 \Prob\!\left\{
 R_1(\widehat h)\le\alpha_1,\quad
 R_2(\widehat h)\le\alpha_2
 \right\}
 \ge(1-\eta_1)(1-\eta_2)\ge1-\delta.
\]
The last bound also holds conditional on the training information alone, including when a class is absent and the algorithm returns only indecisions.
\end{theorem}

The key is to couple all candidate rules to a common reference sequence within each class. The two-tail structure produces independent Bernoulli$(\alpha_i)$ indicators whose mixture likelihood ratios form a single martingale. The calibration statistic of every pair with $R_i(h_{\ell,u})>\alpha_i$ is bounded above by this martingale at the pair's decision count. Thus, if any violating pair passes, the common process must cross $c_i$ within $n_i$ steps. The recursion in Lemma~\ref{lem:fixed-pair-calibration} calculates the probability of this event, accounting for all decision counts and threshold pairs together. Equality concerns all real threshold pairs and one class's check; the selected rule must lie on the calibration grid and pass both checks, so its violation probability can be smaller. Under the same conditioning, the two classwise reference processes are independent. Multiplying their noncrossing probabilities gives the joint guarantee. The detailed proof is given in Appendix Section~\ref{app:calibration-proof}. Appendix Section~\ref{app:general-boundaries} extends the count comparison to other prespecified calibration boundaries.
\endgroup

The guarantee remains valid under label shift: class proportions may differ between training, calibration, and target populations, while the conditional distributions $P_i$ remain fixed \citep{lipton2018labelshift}. Both selective errors depend only on $P_i$, so the same calibration event controls them for every choice of positive target proportions. Proposition~\ref{prop:largest-region} also gives maximal coverage among passing grid rules for every such target population.

Our sampling model assumes that calibration and future feature vectors from class $i$ are independent draws from the same distribution $P_i$. Appendix Section~\ref{app:exchangeability} extends the guarantee to an infinitely exchangeable labeled sequence, using a fixed initial training block and the first prespecified number of subsequent observations from each class for calibration.

\subsection{The cost of calibration}
\label{subsec:calibration-efficiency}

For a given score $s$, the ideal benchmark uses the known class-conditional score distributions to minimize population indecision subject to both error constraints. NPI selects only pairs that pass the calibration tests in \eqref{eq:table-check}, which can exclude pairs satisfying both population constraints. We quantify the resulting additional indecision relative to this ideal benchmark.

As before, we condition on the training information and hold $s$ fixed. Using the coverage $C(h)$ defined in Section~\ref{subsec:errors}, the ideal coverage is
\begin{equation}
 C_s^\star
 =\sup_{\ell\le u}
 \bigl\{C(h_{\ell,u}):R_i(h_{\ell,u})\le\alpha_i,\ i=1,2\bigr\}.
 \label{eq:score-coverage-oracle}
\end{equation}
When this supremum is attained, denote an ideal pair by $(\ell_s^\star,u_s^\star)$ and its rule by $h_s^\star=h_{\ell_s^\star,u_s^\star}$. The difference $C_s^\star-C(\widehat h) =\Gamma(\widehat h)-\Gamma(h_s^\star)$ measures the additional indecisions due to calibration. Let $F_i(z)=P_i\{s(X)\le z\}$, $n_{\min}=\min(n_1,n_2)$, and $c_{\max}=\max(c_1,c_2)$.

\begingroup
\begin{theorem}
\label{thm:coverage-oracle}
Under the assumptions of Theorem~\ref{thm:simultaneous}, suppose an ideal pair $(\ell_s^\star,u_s^\star)$ exists with finite thresholds and $q_i(h_s^\star)>0$, $i=1,2$. Assume that each $F_i$ is continuously differentiable in neighborhoods of both thresholds. Suppose there is a direction $v=(v_\ell,v_u)$ such that $\ell_s^\star+t v_\ell\le u_s^\star+t v_u$ for all sufficiently small $t>0$, and
\begin{equation}
 \left.\frac{d}{dt}
 R_i\bigl(h_{\ell_s^\star+t v_\ell,\,u_s^\star+t v_u}\bigr)
 \right|_{t=0+}<0
 \quad\text{for each }i\text{ with }R_i(h_s^\star)=\alpha_i.
 \label{eq:local-oracle-regularity}
\end{equation}
Let $\widehat h$ be returned by Algorithm~\ref{alg:calibration}. There exist constants $K,c_0>0$, depending only on the fixed score distributions, error targets, and population proportions, such that, for any $\beta\in(0,1-\delta)$ and $n_{\min}\ge2$ satisfying $\log\{n_{\min}c_{\max}/\beta\}/n_{\min}\le c_0$, with probability at least $1-\delta-\beta$ conditional on the training information and calibration labels, the selected rule satisfies $R_i(\widehat h)\le\alpha_i$, $i=1,2$, and
\begin{equation}
 0\le C_s^\star-C(\widehat h)
 \le K\sqrt{\frac{\log\{n_{\min}c_{\max}/\beta\}}
                    {n_{\min}}}.
 \label{eq:coverage-oracle-bound}
\end{equation}
\end{theorem}
\endgroup

The bound quantifies the coverage cost through $c_{\max}$, which is smaller under exact calibration than under Ville's inequality. For fixed failure allowances, the bound is of order $\sqrt{\log n_{\min}/n_{\min}}$. The proof first moves the ideal thresholds slightly to bring both errors below their targets. Condition~\eqref{eq:local-oracle-regularity} ensures that an error margin of size $t$ costs at most $O(t)$ in coverage. Taking $t$ of the order in \eqref{eq:coverage-oracle-bound} makes this nearby rule pass calibration with high probability. Rounding its thresholds to the calibration grid preserves all calibration decisions, and concentration of the classwise score distributions bounds the resulting population coverage loss. The selected decision region contains that of this grid rule. The detailed proof is given in Appendix Section~\ref{app:coverage-proof}.

\begin{remark}
\label{rem:local-oracle-regularity}
Condition~\eqref{eq:local-oracle-regularity} requires a common movement of the thresholds that reduces every error rate at its target. It holds trivially when neither constraint is active at the ideal pair. For distinct ideal thresholds, linear independence of the gradients of the active error rates is sufficient. The condition is sufficient rather than necessary, but some control of the coverage lost in creating an error margin is essential for the stated rate. Smooth score distributions alone do not suffice: if both classes have the same uniform score distribution and $\alpha_1=\alpha_2=1/2$, the ideal rule classifies everyone, yet no rule with positive coverage satisfies both constraints strictly. The example in Appendix Section~\ref{app:local-regularity-example} shows that NPI then returns only indecisions with probability at least $1-\delta$, regardless of sample size. These local assumptions concern the coverage bound; Theorem~\ref{thm:simultaneous} requires none of them.
\end{remark}

\section{Simulations}
\label{sec:simulations}

We evaluate joint error control and indecision rates in two settings: Gaussian classes with varying overlap and a nonlinear classification problem with concentric rings. We compare NPI with NP umbrella \citep{tong2018neyman}, the method of \citet{ndaoud2026ask}, and two implementations of LTT \citep{angelopoulos2025learn}.

\subsection{Simulation design and competing methods}
\label{subsec:simulation-design}

At each separation level, we conduct 5,000 independent repetitions. Each repetition uses independent training and calibration samples, each balanced between the two classes. Sample sizes are specified below, and the target population has class proportions $\pi_1=\pi_2=0.5$. All five procedures share the same fitted score and samples within each repetition.

We set $\alpha_1=\alpha_2=0.10$. NPI uses a joint failure tolerance $\delta=0.05$, with $\eta_1=\eta_2=1-\sqrt{0.95}\approx0.02532$ and the exact critical values in \eqref{eq:exact-cutoff}. We implement NP umbrella using the calibration-rank rule of \texttt{nproc} \citep{tong2018neyman}, applied to the shared score. It controls type I error with failure tolerance $0.05$; it makes no indecisions and imposes no type II error constraint. At its threshold $t_{\mathrm{NP}}$, scores below $t_{\mathrm{NP}}$ are assigned to class 2 and scores at or above $t_{\mathrm{NP}}$ to class 1, retaining the \texttt{nproc} convention when scores are tied. The method of \citet{ndaoud2026ask} uses their Algorithm~2 with both empirical error targets set to $0.10$. If no empirically feasible rule is found, we return indecisions everywhere.

The LTT framework can be adapted to our problem by testing both selective error constraints for each candidate threshold pair. We consider LTT Holm and LTT sequence, which use Holm's adjustment \citep{holm1979simple} and fixed-sequence testing, respectively, at family-wise error level $\delta=0.05$. Implementation details are given in Appendix Section~\ref{app:ltt-implementation}.

\paragraph{Evaluation.}
Conditional on the fitted score and selected thresholds, we calculate the population selective errors and indecision probability directly from the generating distributions. No test sample is used. For a two-threshold rule, the classwise score distributions $F_i$ give decision probabilities $q_i=F_i(\ell)+1-F_i(u)$, the two errors are $R_1=F_1(\ell)/q_1$ and $R_2=\{1-F_2(u)\}/q_2$, and the indecision probability is $1-\sum_i\pi_iq_i$. Joint violation frequencies count repetitions in which either population error exceeds its target. All repetitions enter the error distributions, their quantiles, the joint violation frequencies, and the mean indecision summaries. When a rule makes no decisions in a class, its error is zero under the convention $0/0=0$ in \eqref{eq:risks}.

A population oracle provides a reference by minimizing indecisions subject to both error constraints using the known class distributions. It has no sampling variability. Differences from the oracle can reflect score approximation, finite-sample calibration, and restrictions on the candidate rules.

\subsection{Gaussian separation}
\label{subsec:gaussian-simulation}

We generate
\[
 X\mid Y=1\sim N(d/2,1),
 \qquad
 X\mid Y=2\sim N(-d/2,1),
\]
and display results for $d\in\{1.5,2.0,2.5,3.0,3.5\}$. The training and calibration samples each contain 1,000 observations, with 500 per class. We fit a linear discriminant analysis (LDA) score \citep{hastie2009elements,James2013} and evaluate population errors using normal probabilities. The oracle uses the likelihood ratio $L(x)=\exp(dx)$.

\begin{figure}[!t]
    \centering
    \resizebox{\textwidth}{!}{%
        \input{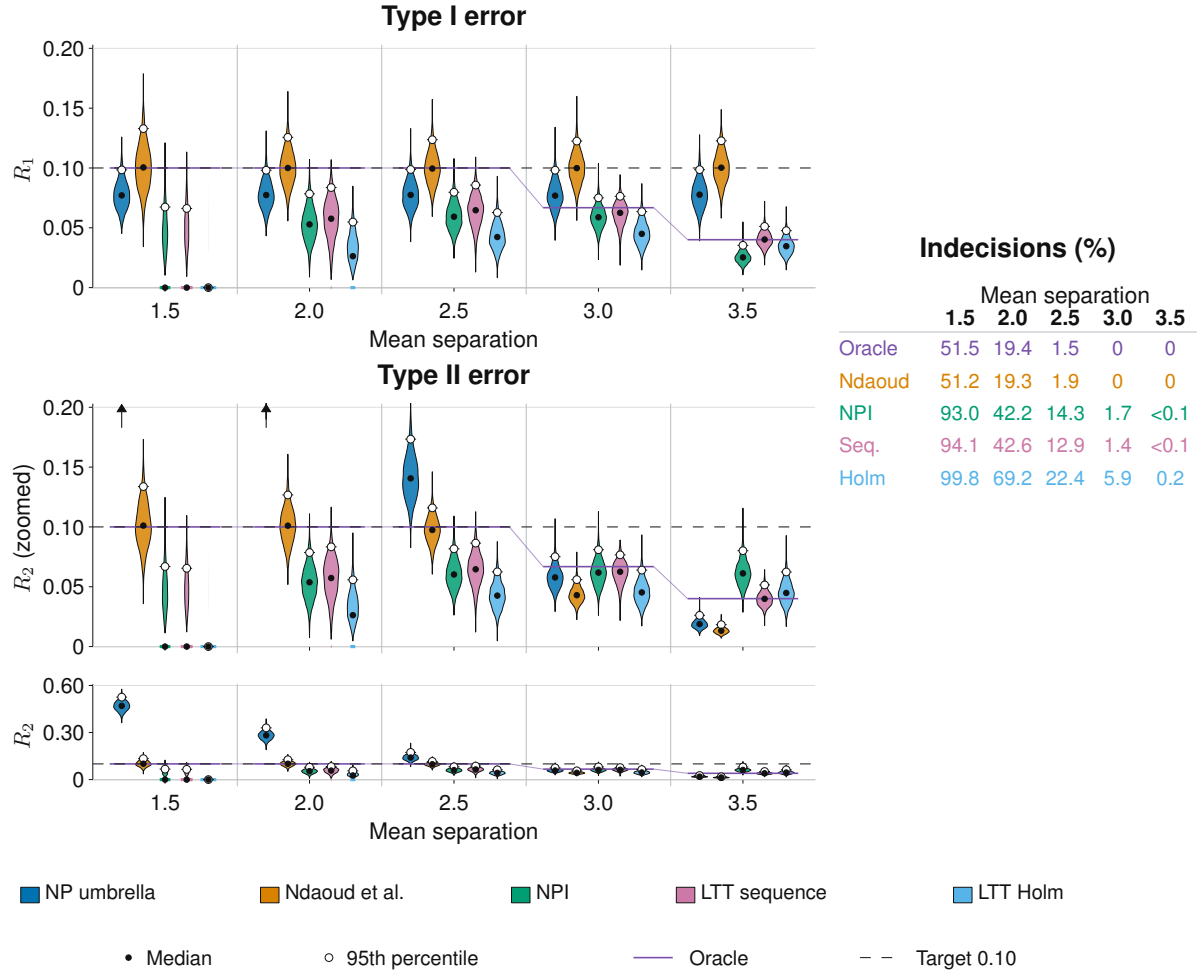}%
    }
    \caption{\small Gaussian simulation with shared LDA scores and 5,000 repetitions per separation. Type I errors appear above type II errors; the main type II panel is enlarged, with a smaller full-range panel below. Zero errors are shown as horizontal bars; bar lengths and nonzero violin widths are scaled by their respective proportions of repetitions. Filled and open markers indicate medians and 95th percentiles; arrows mark off-scale values. Dashed lines mark the $10\%$ target and purple segments show the population oracle. The table reports population indecision percentages, averaged over all repetitions for the fitted procedures. NP umbrella's zero-indecision row is omitted.
    }
    \label{fig:np-gaussian-separation}
\end{figure}

\begin{figure}[!t]
    \centering
    \resizebox{\textwidth}{!}{%
        \input{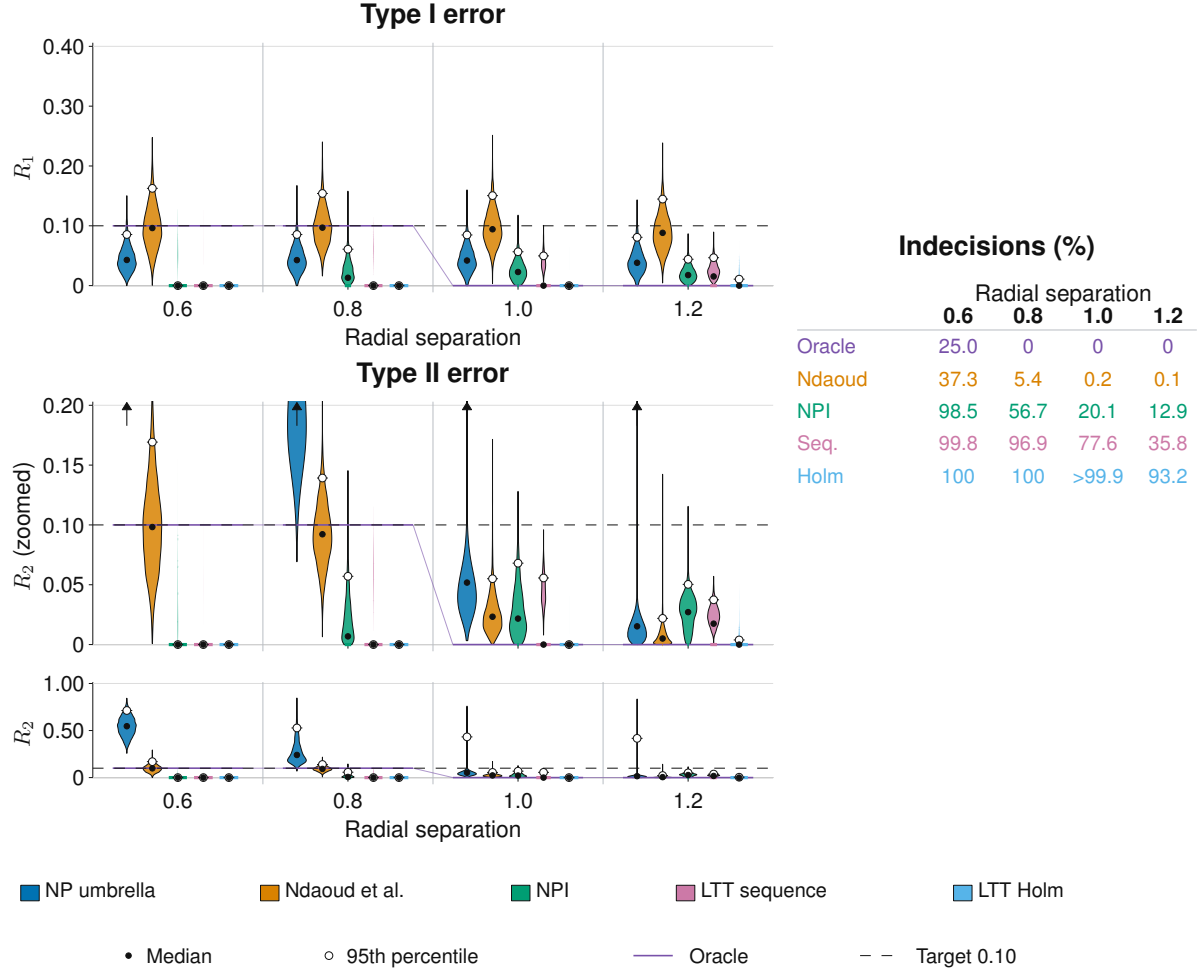}%
    }
    \caption{\small Concentric-ring simulation with shared AdaBoost scores, 200 training and 200 calibration observations, and 5,000 repetitions per separation. Plotting conventions and indecision summaries follow Figure~\ref{fig:np-gaussian-separation}.
    }
    \label{fig:np-rings-separation}
\end{figure}

Figure~\ref{fig:np-gaussian-separation} shows comparable indecision rates for NPI and LTT sequence, with greater indecision for LTT Holm. At $d=2$, their mean indecision percentages are $42.2\%$, $42.6\%$, and $69.2\%$, respectively, compared with $19.4\%$ for the oracle. At $d=3$, these percentages fall to $1.7\%$, $1.4\%$, and $5.9\%$, while the oracle classifies every observation.

This Gaussian setting favors LTT sequence: the correctly specified LDA score closely follows the likelihood-ratio ordering, supporting an informative training-based ordering of threshold pairs. With one starting point, the sequence tests each pair at level $\delta$ and stops at the first nonrejection \citep{angelopoulos2025learn}. Holm's adjustment imposes smaller testing levels across the candidate collection, making the calibration checks more demanding and increasing indecision. NPI requires no such multiplicity adjustment (Theorem~\ref{thm:simultaneous}).

Across the displayed separations, the largest observed joint violation frequencies are $0.50\%$ for NPI and $0.66\%$ for LTT sequence; no joint violations are observed for LTT Holm. These frequencies are below the prescribed $5\%$ failure tolerance. The method of \citet{ndaoud2026ask} makes fewer indecisions, but its joint violation frequencies range from $49.5\%$ to $77.3\%$. NP umbrella controls type I error, while its type II error exceeds the target in every repetition at $d=1.5$ and $d=2$.

The exact crossing probability bounds the event that some violating pair passes one class's check. This event includes pairs outside the search grid and pairs that fail the other class's check. At $d=2$, this bound is $2.52\%$ per class, whereas the selected rule's observed type I and type II violation frequencies are $0.10\%$ and $0.22\%$, respectively. For both NPI and LTT, passing calibration requires finite-sample evidence that the errors are below their $10\%$ targets, favoring rules with a margin below those targets.

\subsection{Nonlinear separation}
\label{subsec:nonlinear-simulation}

We generate concentric rings through $X=R(\cos\Theta,\sin\Theta)$, where $\Theta\sim\operatorname{Unif}(0,2\pi)$ is independent of $R$, and
\[
 R\mid Y=1\sim\operatorname{Unif}(0.5,1.5),
 \qquad
 R\mid Y=2\sim\operatorname{Unif}(0.5+d,1.5+d).
\]
We consider $d\in\{0.6,0.8,1.0,1.2\}$, with sampling uniform in radius. The training and calibration samples each contain 200 observations, with 100 per class. We fit AdaBoost with at most 100 decision stumps \citep{freund1997decision} using the Cartesian coordinates. Population probabilities are obtained by analytic integration over the rectangular regions on which the fitted score is constant. The oracle classifies class-exclusive regions without error and permits randomization within the overlap.

NPI achieves substantially higher decision coverage than both LTT implementations at $d\ge0.8$ (Figure~\ref{fig:np-rings-separation}). At $d=1$, mean indecision is $20.1\%$ for NPI, $77.6\%$ for LTT sequence, and $99.9\%$ for LTT Holm. At $d=1.2$, these percentages decrease to $12.9\%$, $35.8\%$, and $93.2\%$, respectively. At the smallest separation, $d=0.6$, NPI and both LTT implementations make few decisions.

This setting illustrates the value of selecting thresholds without a prescribed testing order. Learning the nonlinear score is more difficult, and training errors can be overly optimistic about which threshold pairs will pass independent calibration. LTT sequence stops at the first failed calibration check, potentially leaving untested pairs that would pass. NPI evaluates the full calibration grid without relying on a training-based ordering, allowing it to identify passing pairs with higher coverage. LTT Holm also searches its full candidate collection, but its multiplicity adjustment makes the calibration checks more demanding.

The coverage advantage is accompanied by joint error control: the largest observed joint violation frequency is $0.54\%$ for NPI, compared with $0.20\%$ for LTT sequence and zero for LTT Holm, all below the $5\%$ failure tolerance. The method of \citet{ndaoud2026ask} makes more decisions, but its joint violation frequencies range from $35.5\%$ to $72.9\%$. Together with Section~\ref{subsec:gaussian-simulation}, these results demonstrate comparable coverage for NPI and LTT sequence in the Gaussian setting, with substantial gains for NPI in this more challenging nonlinear setting.

\section{Real data}
\label{sec:real_data}

We study recidivism prediction using COMPAS data in this section and credit-card default prediction using data from the University of California, Irvine (UCI) Machine Learning Repository in Appendix Section~\ref{app:credit-application}. We compare NPI with NP umbrella, the method of \citet{ndaoud2026ask}, and the two LTT implementations used in Section~\ref{sec:simulations}. We examine how their selective errors and decision coverage change as the common target $\alpha_1=\alpha_2=\alpha$ varies.

For each data set, we randomly partition all observations into training, calibration, and test samples 1,000 times, without balancing the classes. Within each partition, data-dependent preprocessing and score fitting use only the training sample. All methods and targets share this fitted score and the same calibration and test samples. Larger scores favor class 1; the learners and sample sizes are specified with each application. Model and sample-size choices were informed by exploratory analyses of these data sets and held fixed across the reported partitions.

At each target, NPI and both LTT implementations use a joint failure allowance $\delta=0.05$, with NPI assigning $1-\sqrt{0.95}$ to each class. NPI's critical values depend on the calibration class sizes $n_1,n_2$ in each partition. NP umbrella uses a failure tolerance of $0.05$ for its type I error constraint. The test sample is reserved for evaluation.

Within each actual class, we calculate the fraction of incorrect classifications among test observations receiving definitive decisions, using $0/0=0$ when no test decisions are made in that class. All 1,000 partitions enter the error distributions, their quantiles, and the mean test indecision rate. These empirical errors include test-sample variation, so their quantiles do not directly assess the joint population guarantee in \eqref{eq:joint-control}.

\begin{figure}[!t]
    \centering
    \resizebox{\textwidth}{!}{%
        \input{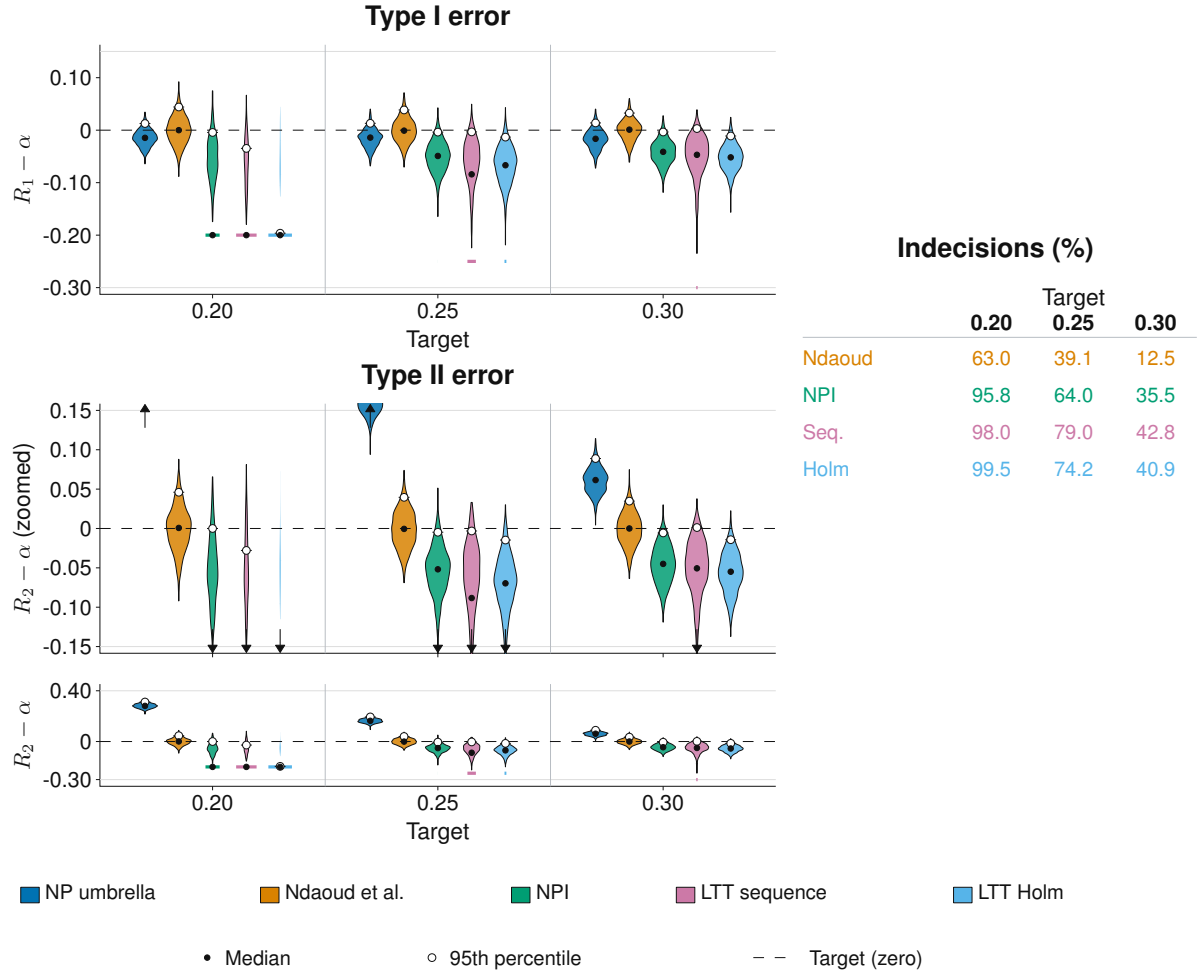}%
    }
    \caption{\small COMPAS results over 1,000 partitions, using logistic regression with 500 training, 4,000 calibration, and 1,672 test observations. Violins show held-out selective error minus $\alpha$, with type I above type II; the latter includes a smaller full-range panel. Zero-error bars and positive-error violin widths are scaled by their respective proportions of partitions. Filled and open markers indicate medians and 95th percentiles; arrows mark off-scale values. The table gives mean test indecision percentages, omitting NP umbrella's zero-indecision row.
    }
    \label{fig:compas_alpha}
\end{figure}

\subsection{Recidivism prediction}
\label{subsec:compas-application}

The COMPAS (Correctional Offender Management Profiling for Alternative Sanctions) data were compiled by ProPublica to study recidivism risk assessment in Broward County, Florida \citep{Angetal16-COMPAS1}. Applying ProPublica's screening criteria \citep{larson2016compasanalysis} yields 6,172 individuals. Class 1 comprises 3,363 individuals without recorded recidivism within two years; class 2 comprises 2,809 with recorded recidivism.

Each partition assigns 500 individuals to training, 4,000 to calibration, and 1,672 to testing. We fit an unpenalized logistic regression using age, the logarithm of one plus the prior-offense count, and the supplied COMPAS decile score. The estimated probability of no recorded recidivism provides the shared classification score. We consider $\alpha\in\{0.20,0.25,0.30\}$.

Selective type I error measures the fraction incorrectly classified as recidivists among class 1 individuals receiving a definitive classification. Selective type II error is the corresponding fraction incorrectly classified as nonrecidivists among class 2 individuals.

Figure~\ref{fig:compas_alpha} shows that NPI classifies more test observations than either LTT implementation at each displayed target. At $\alpha=0.25$, mean test indecision is $64.0\%$, $79.0\%$, and $74.2\%$ for NPI, LTT sequence, and LTT Holm, respectively; at $\alpha=0.30$, the corresponding values are $35.5\%$, $42.8\%$, and $40.9\%$. The 95th percentiles of NPI's selective type I and type II test errors are $24.6\%$ and $24.5\%$ at $\alpha=0.25$, and $29.7\%$ and $29.4\%$ at $\alpha=0.30$.

LTT Holm's test-error 95th percentiles also remain below both targets. For LTT sequence, they are $24.7\%$ for both errors at $\alpha=0.25$, and $30.3\%$ and $30.1\%$ at $\alpha=0.30$. LTT sequence returns only indecisions in 342 and 56 partitions at these targets, compared with one and zero for NPI.

The $20\%$ target permits few decisions: mean test indecision is $95.8\%$ for NPI, $98.0\%$ for LTT sequence, and $99.5\%$ for LTT Holm. The 95th percentiles of NPI's two test errors are $19.6\%$ and $20.0\%$. The 570 partitions in which NPI returns only indecisions contribute zeros to both of its error distributions.

In this application and the credit-card analysis in Appendix Section~\ref{app:credit-application}, the method of \citet{ndaoud2026ask} makes more decisions than NPI, while the 95th percentiles of both of its selective test errors exceed the displayed targets. NP umbrella classifies every observation, with type II error quantiles above the targets.

\section{Discussion}
\label{sec:discussion}

The NPI framework combines separate control of selective type I and type II errors with an explicit objective of reducing indecisions. Score quality affects coverage, although the joint finite-sample guarantee holds for any independently trained score.

The oracle rule of Section~\ref{subsec:oracle} and the ideal score-specific benchmark of Section~\ref{subsec:calibration-efficiency} separate the sources of indecision. Class overlap can limit coverage even when the likelihood ratio is known, and an estimated score may further restrict the coverage attainable by thresholding. Theorem~\ref{thm:coverage-oracle} bounds the additional loss due to calibration under local regularity conditions. Low coverage can therefore reflect either limited separation or insufficient calibration data. Improving the score and collecting more calibration observations address different sources of this loss. The dependence on the smaller classwise sample size highlights the importance of calibration data from both classes. How to allocate a fixed data budget between score learning and calibration, particularly when one class is scarce, remains an important question.

The present theory concerns binary classification and assumes independent sampling within each class, with an exchangeable extension given in the appendix. The common calibration event covers many rules that are never selected, so the final rule's error guarantee can be conservative. Appendix Section~\ref{app:general-boundaries} allows other calibration boundaries within the same count-comparison argument. Choosing these boundaries to improve coverage is a further direction.

Another extension concerns the choice of the score and error targets, which our theory requires to be fixed before calibration. Choosing among several fitted scores or target pairs using their calibration results introduces an additional selection problem. Simultaneous calibration across these choices would allow error--coverage comparisons to guide the final rule. Developing such a procedure while limiting additional indecisions is a useful direction for future research.

\bibliographystyle{plainnat}
\bibliography{myrefs_revised}

\clearpage
\appendix
\renewcommand{\theequation}{\thesection.\arabic{equation}}
\setcounter{equation}{0}
\begin{center}
  {\Large\bfseries Online Supplementary Material for\par}
  \medskip
  {\large\papertitle\par}
\end{center}
\bigskip

Section~\ref{app:proofs} gives proofs of the main results. Section~\ref{sec:extensions} presents additional theoretical results, including extensions to infinite exchangeability and other calibration boundaries. Section~\ref{app:calibration-tables} describes the LTT implementations, and Section~\ref{app:credit-application} reports the credit-card default analysis.

\section{Proofs of the main results}
\label{app:proofs}

\subsection{Proof of Theorem~\ref{thm:oracle}}
\label{app:oracle-proof}

\begingroup\color{black}
\begin{proof}
For any feasible classifier, we construct a two-tail rule with the same coverage and no larger selective errors. We then show that the best such rule is attained by optimizing over the probabilities of assigning each class.

\paragraph{Step 1: Place the class assignments in two ordered tails.}
Let
\[
 \vartheta(x)=\Prob(Y=1\mid X=x)
 =\frac{\pi_1f_1(x)}{\pi_1f_1(x)+\pi_2f_2(x)}.
\]
Its value where both densities vanish is immaterial. Throughout the proof, expectations are taken under the marginal distribution of $X$. For a randomized rule, let $p_j(x)$ be the probability of assigning action $j$ at $x$, $j=0,1,2$. Thus $p_j\ge0$ and $p_0+p_1+p_2=1$. Fix any feasible rule and write
\[
 a_j=\E\{p_j(X)\},\qquad j=1,2.
\]
Its coverage is $a_1+a_2\le1$. We preserve these two assignment probabilities while placing class 1 on the largest posterior scores and class 2 on the smallest.

Let $G(z)=\Prob\{\vartheta(X)\le z\}$ and, for $v\in(0,1)$, define $Q(v)=\inf\{z\in[0,1]:G(z)\ge v\}$. Set $Q(0)=0$ and $Q(1)=1$, and choose
\[
 a=Q(a_2),\qquad b=Q(1-a_1).
\]
The inequality $a_2\le1-a_1$ gives $a\le b$. Define the new action probabilities by
\[
 \begin{aligned}
 \widetilde p_2(x;a_2)
 &=\ind\{\vartheta(x)<a\}
   +\gamma_2\ind\{\vartheta(x)=a\},\\
 \widetilde p_1(x;a_1)
 &=\ind\{\vartheta(x)>b\}
   +\gamma_1\ind\{\vartheta(x)=b\},
 \end{aligned}
\]
where
\[
 \gamma_2=\frac{a_2-\Prob\{\vartheta(X)<a\}}
                 {\Prob\{\vartheta(X)=a\}},
 \qquad
 \gamma_1=\frac{a_1-\Prob\{\vartheta(X)>b\}}
                 {\Prob\{\vartheta(X)=b\}}.
\]
Take the corresponding $\gamma_j=0$ when its denominator is zero. The quantile definitions give $\Prob\{\vartheta(X)<a\}\le a_2\le\Prob\{\vartheta(X)\le a\}$ and $\Prob\{\vartheta(X)>b\}\le a_1\le\Prob\{\vartheta(X)\ge b\}$. Hence $\gamma_j\in[0,1]$ and $\E\{\widetilde p_j(X;a_j)\}=a_j$, $j=1,2$.

When $a<b$, the two class assignments have disjoint supports. If $a=b$ and the common boundary has probability $d=\Prob\{\vartheta(X)=a\}>0$, then
\[
 \gamma_1+\gamma_2
 =\frac{a_1+a_2-1+d}{d}\le1.
\]
Thus the assignments are compatible at a shared boundary as well. A boundary of zero probability does not affect the rule. Setting $\widetilde p_0=1-\widetilde p_1-\widetilde p_2$ therefore gives a classifier with coverage $a_1+a_2$.

\paragraph{Step 2: Compare the classwise error probabilities.}
For this comparison, suppress $a_j$ in $\widetilde p_j$. The upper-tail assignment satisfies
\[
 (\vartheta-b)(\widetilde p_1-p_1)\ge0.
\]
Above $b$, the new class 1 assignment probability is one; below $b$, it is zero; at $b$, the product vanishes. Taking expectations and using $\E\{\widetilde p_1(X)\}=\E\{p_1(X)\}$ gives
\[
 \E\{\vartheta(X)\widetilde p_1(X)\}
 \ge\E\{\vartheta(X)p_1(X)\}.
\]
The lower-tail assignment similarly satisfies $(\vartheta-a)(\widetilde p_2-p_2)\le0$, yielding
\[
 \E\{\vartheta(X)\widetilde p_2(X)\}
 \le\E\{\vartheta(X)p_2(X)\}.
\]
Write $P_i g=\E\{g(X)\mid Y=i\}$. By conditioning on $X$,
\[
 P_1g=\frac{\E\{\vartheta(X)g(X)\}}{\pi_1},
 \qquad
 P_2g=\frac{\E\{(1-\vartheta(X))g(X)\}}{\pi_2}.
\]
These identities and the preserved assignment probabilities imply
\[
 \begin{aligned}
 P_1\widetilde p_1&\ge P_1p_1,
 &P_2\widetilde p_1&\le P_2p_1,\\
 P_1\widetilde p_2&\le P_1p_2,
 &P_2\widetilde p_2&\ge P_2p_2.
 \end{aligned}
\]
Within each class, the probability of a correct decision increases or stays the same, while the probability of an incorrect decision decreases or stays the same. Consequently, neither selective error increases.

For the limiting argument below, express the error constraints in terms of these decision probabilities. With $0/0=0$, $R_i(h)\le\alpha_i$, $i=1,2$, is equivalent to
\begin{equation}
 \begin{split}
 (1-\alpha_1)P_1p_2-\alpha_1P_1p_1&\le0,\\
 (1-\alpha_2)P_2p_1-\alpha_2P_2p_2&\le0.
 \end{split}
 \label{eq:linear}
\end{equation}
The four inequalities above show directly that replacing $p_j$ by $\widetilde p_j$ preserves both constraints. Thus every feasible classifier has a feasible two-tail counterpart with the same coverage.

\paragraph{Step 3: Show that an optimal pair of assignment probabilities exists.}
The tail construction in Step 1 is valid for every pair satisfying
\[
 a_1\ge0,\qquad a_2\ge0,\qquad a_1+a_2\le1.
\]
Using assignment probabilities lets us compare rules even when the posterior score equals a threshold with positive probability. Increasing $a_j$ enlarges the corresponding tail; if its threshold stays at a value with positive probability, the boundary assignment probability increases. Hence $\widetilde p_j(x;a_j)$ is nondecreasing in $a_j$ for each $x$. Since its expectation is $a_j$,
\[
 \E\bigl|\widetilde p_j(X;a_j)-\widetilde p_j(X;a_j')\bigr|
 =|a_j-a_j'|.
\]
For each class, $\pi_iP_i|g|\le\E|g(X)|$. Therefore
\[
 \left|P_i\widetilde p_j(\cdot;a_j)
       -P_i\widetilde p_j(\cdot;a_j')\right|
 \le\frac{|a_j-a_j'|}{\pi_i},\qquad i,j=1,2.
\]
Thus small changes in the assignment probabilities produce small changes in all four classwise decision probabilities.

Let $V$ be the supremum of coverage over the feasible pairs. The pair $(0,0)$ is feasible, so $V\in[0,1]$. Choose feasible pairs whose coverages approach $V$. Since both coordinates lie in $[0,1]$, a subsequence converges to a pair $(a_1^\star,a_2^\star)$ with $a_1^\star+a_2^\star\le1$. The bound above ensures convergence of the four classwise decision probabilities. Taking limits in \eqref{eq:linear} shows that the rule indexed by $(a_1^\star,a_2^\star)$ remains feasible. Its coverage is $a_1^\star+a_2^\star=V$. By Step 2, this is also the largest coverage among all feasible classifiers.

\paragraph{Step 4: Express the optimal rule using the likelihood ratio.}
For the optimal pair, let $a^\star=Q(a_2^\star)$ and $b^\star=Q(1-a_1^\star)$. The posterior score satisfies
\[
 \vartheta(x)=\frac{\pi_1L(x)}{\pi_1L(x)+\pi_2},
\]
with posterior value zero when $f_1=0<f_2$ and one when $f_2=0<f_1$; set $L=0$ and $L=+\infty$, respectively, on these two sets. This transformation is strictly increasing. Thus
\[
 \ell^{\mathrm{OR}}=\frac{\pi_2a^\star}{\pi_1(1-a^\star)},
 \qquad
 u^{\mathrm{OR}}=\frac{\pi_2b^\star}{\pi_1(1-b^\star)},
\]
with value $+\infty$ when the corresponding posterior threshold is one, give the same decisions as the optimal two-tail rule. They satisfy $0\le\ell^{\mathrm{OR}}\le u^{\mathrm{OR}}\le+\infty$, with the boundary assignments permitted in \eqref{eq:oracle}. This proves that the minimum in \eqref{eq:oracle-cutoffs} equals the global minimum in \eqref{eq:population}.

If some rule without indecisions satisfies both constraints, then $V=1$ and $a_2^\star=1-a_1^\star$. Hence $a^\star=b^\star$, giving coincident likelihood-ratio thresholds. At a common boundary of positive probability, $\gamma_1+\gamma_2=1$, so no indecisions remain there.
\end{proof}
\endgroup

\subsection{Proof of Lemma~\ref{lem:fixed-pair-calibration}}
\label{app:evidence-interpretation}

\begin{proof}
We first verify that the reference process is a martingale and derive its crossing recursion. We then establish the existence of the exact critical value and compare a fixed candidate rule with this reference process.

\paragraph{Step 1: Verify the reference martingale.}
For $r\in(0,\alpha)$ and nonnegative counts $e,k$, define the binomial likelihood ratio
\begin{equation}
 \mathscr L_{\alpha,r}(e,k)
 =\left(\frac r\alpha\right)^e
  \left(\frac{1-r}{1-\alpha}\right)^k.
 \label{eq:calibration-lr}
\end{equation}
For the independent Bernoulli$(\alpha)$ reference variables, write $W_m=\sum_{j=1}^m\xi_j$ and
\begin{equation}
 \mathscr L_m(r)=\mathscr L_{\alpha,r}(W_m,m-W_m),
 \qquad \mathscr L_0(r)=1.
 \label{eq:reference-likelihood-process}
\end{equation}
Conditional on $\xi_1,\ldots,\xi_{m-1}$, the expected multiplier from the next observation is
\[
 \alpha\frac r\alpha
 +(1-\alpha)\frac{1-r}{1-\alpha}=1.
\]
Thus $\mathscr L_m(r)$ is a nonnegative mean-one martingale. Integrating these nonnegative processes over $r$ gives the same property for
\begin{equation}
 T_m=\frac1\alpha\int_0^\alpha\mathscr L_m(r)\,dr
     =B_\alpha(W_m,m-W_m),\qquad T_0=1.
 \label{eq:reference-mixture-process}
\end{equation}

Fix $c\ge1$, and stop at the first $m\le n$ with $T_m>c$, or at $n$ if there is no crossing. This bounded stopping time $\tau$ satisfies $\E(T_\tau)=1$. On the crossing event, $T_\tau>c$, so nonnegativity gives $\psi_{n,\alpha}(c)\le1/c$. The same argument with a non-strict crossing gives
\[
 \Prob\!\left\{\max_{1\le m\le n}T_m\ge a\right\}\le1/a,
 \qquad a>1.
\]

\paragraph{Step 2: Compute the crossing probability.}
For fixed $m$, replacing a correct decision by an error multiplies the integrand in $B_\alpha(e,m-e)$ by
\[
 \frac{r(1-\alpha)}{\alpha(1-r)}<1,
 \qquad 0<r<\alpha.
\]
Consequently, $B_\alpha(e,m-e)$ is strictly decreasing in $e$, and $T_m>c$ is equivalent to $W_m\le b_{\alpha,c}(m)$. Let $p_m(w)$ be the probability that $W_m=w$ and the process has not crossed by time $m$. Initially $p_0(0)=1$, with all other entries zero. A surviving path reaches $w$ at the next step from $w-1$ with probability $\alpha$, or from $w$ with probability $1-\alpha$. It remains a surviving path exactly when $w>b_{\alpha,c}(m)$. This proves the recursion in \eqref{eq:crossing-recursion}, with entries outside $0\le w\le m$ taken to be zero. Summing the surviving probabilities at time $n$ yields $\psi_{n,\alpha}(c)=1-\sum_{w=0}^n p_n(w)$.

\paragraph{Step 3: Establish the exact critical value.}
The reference maximum takes finitely many values. As $c$ increases, the crossing probability changes only when $c$ reaches one of these values, which the strict check excludes at equality. Above the largest value, the probability is zero. Hence the first admissible critical value in \eqref{eq:exact-cutoff} is attained. To compare it with Ville's critical value, fix $\eta\in(0,1)$ and take $a=1/\eta>1$. Finiteness of the support permits a value $c\in[1,a)$ such that
\[
 \left\{\max_{1\le m\le n}T_m>c\right\}
 =\left\{\max_{1\le m\le n}T_m\ge a\right\}.
\]
The non-strict crossing bound in Step 1 makes this $c$ admissible. Thus $1\le c_{n,\alpha}(\eta)<1/\eta$. In particular, every count configuration satisfying $B_\alpha(e,k)\ge1/\eta$ also satisfies the exact check.

\paragraph{Step 4: Bound the error for a fixed candidate pair.}
Condition on the training information and calibration labels, so the score, prespecified pair, and class counts are fixed. For class $i$, abbreviate $\alpha=\alpha_i$, $\theta=R_i(h_{\ell,u})$, and $q=P_i\{h_{\ell,u}(X)\ne0\}$. Under the null $\theta\ge\alpha>0$, $q>0$ because a rule making no decisions has error zero by convention. The calibration observations independently produce an erroneous decision, a correct decision, or an indecision with probabilities $q\theta$, $q(1-\theta)$, and $1-q$. Therefore
\[
 m_i\sim\operatorname{Bin}(n_i,q),\qquad
 e_i\mid m_i=m\sim\operatorname{Bin}(m,\theta).
\]
Generate Bernoulli variables at probabilities $\alpha$ and $\theta$ using the same independent uniform variables $A_j$: $\ind\{A_j\le\alpha\}\le\ind\{A_j\le\theta\}$ for each $j$. Their sums therefore satisfy the same ordering. Since the statistic decreases with the error count, for $m\ge1$ and $c\ge1$,
\[
 \begin{aligned}
 \Prob\!\left\{B_\alpha(e_i,m_i-e_i)>c\mid m_i=m\right\}
 &\le\Prob\{B_\alpha(W_m,m-W_m)>c\}\\
 &\le\psi_{n_i,\alpha}(c).
 \end{aligned}
\]
When $m=0$, the event is impossible because $B_\alpha(0,0)=1$. Averaging over $m_i$ proves the fixed-pair probability bound.
\end{proof}

\subsection{Proof of Proposition~\ref{prop:largest-region}}
\label{app:largest-region-proof}

\begin{proof}
Fix the fitted score and calibration observations. We first combine the two extreme tails from passing pairs and establish maximality of the resulting decision region. We then compare exact calibration with Ville's critical values.

\paragraph{Step 1: Verify that the selected pair passes calibration.}
Write $k_i=m_i-e_i$ for the number of correct decisions in class $i$. The statistic $B_{\alpha_i}(e_i,k_i)$ decreases when an error is added and increases when a correct decision is added. Thus a passing check remains satisfied if its error count does not increase and its correct count does not decrease.

Because $\mathcal A_{\mathcal G}$ is finite and nonempty, there is a passing pair $(\ell_a,u_a)$ with $\ell_a=\ell_{\max}$ and another $(\ell_b,u_b)$ with $u_b=u_{\min}$. Suppose first that $\ell_{\max}\le u_{\min}$. For class 1, the pair $(\ell_{\max},u_{\min})$ has the same error count as $(\ell_a,u_a)$: both assign scores at most $\ell_{\max}$ to class 2. Since $u_{\min}\le u_a$, its correct count is at least as large. The class-1 check therefore passes. For class 2, compare with $(\ell_b,u_b)$. The error count is unchanged because both upper thresholds equal $u_{\min}$, and the correct count can only increase because $\ell_{\max}\ge\ell_b$. Thus the class-2 check also passes.

If $\ell_{\max}>u_{\min}$, take any $t\in[u_{\min},\ell_{\max}]$. Relative to $(\ell_a,u_a)$, the diagonal pair $(t,t)$ has no more class-1 errors and at least as many class-1 correct decisions, because $t\le\ell_a\le u_a$. Relative to $(\ell_b,u_b)$, it has no more class-2 errors and at least as many class-2 correct decisions, because $\ell_b\le u_b\le t$. Both checks therefore pass. Taking $t=u_{\min}$ gives the selected pair. In both cases the selected thresholds belong to $\mathcal G$, so this pair is an available passing candidate.

\paragraph{Step 2: Show that the selected decision region is the largest.}
When $\ell_{\max}\le u_{\min}$, every passing pair satisfies $\ell\le\ell_{\max}$ and $u\ge u_{\min}$. Consequently,
\[
 (-\infty,\ell]\cup(u,+\infty)
 \subseteq
 (-\infty,\ell_{\max}]\cup(u_{\min},+\infty).
\]
When $\ell_{\max}>u_{\min}$, the selected pair is diagonal and classifies every real score, so the inclusion is immediate. Composing with $s$ proves pointwise inclusion of decision regions in feature space. Integrating their indicators under any distribution proves the asserted coverage maximality.

\paragraph{Step 3: Compare with Ville's critical values.}
Keep the score, calibration observations, targets, grid, and allowances fixed. Lemma~\ref{lem:fixed-pair-calibration} gives $c_i<1/\eta_i$. Every pair satisfying $B_{\alpha_i}(e_i,k_i)\ge1/\eta_i$ for both classes therefore passes the exact checks. By Step 2, its decision region is contained in that of $\widehat h$. An all-indecision rule has an empty decision region, so the inclusion also holds for the fallback rule. Integrating this pointwise inclusion gives no larger indecision probability under any score distribution; averaging it over a fixed test sample gives the corresponding empirical comparison.
\end{proof}

\subsection{Proof of Theorem~\ref{thm:simultaneous}}
\label{app:calibration-proof}

\begin{proof}

We condition throughout on the training data, any independent training randomization, and the calibration labels. First suppose both class counts are positive. The calibration features are then independent, with distribution $P_i$ within class $i$. The score $s$, sample sizes $n_i$, error targets $\alpha_i$, and classwise critical values $c_i$ are therefore fixed. All probabilities and expectations in the proof are conditional on this information. For each class, we construct one Bernoulli reference sequence whose mixture likelihood ratio dominates the calibration statistic of every violating threshold pair at its decision count. The exact crossing probability then bounds the chance that any such pair passes its classwise check. We also show that this bound is attained when the classwise score distribution is continuous, and finally combine the two classwise bounds.

\paragraph{Step 1: Identify the classwise events to be controlled.}
Use the decision and error counts $m_i(\ell,u)$ and $e_i(\ell,u)$ from \eqref{eq:calibration-counts}, and write $k_i(\ell,u)=m_i(\ell,u)-e_i(\ell,u)$ for the number of correct decisions. For any fixed $c\ge1$, define
\[
 \mathcal V_i(c)=\left\{
 \begin{array}{l}
 \text{there exist }-\infty\le\ell\le u\le+\infty\text{ such that}\\[2pt]
 R_i(h_{\ell,u})>\alpha_i
 \quad\text{and}\quad
 B_{\alpha_i}\{e_i(\ell,u),k_i(\ell,u)\}>c
 \end{array}\right\}.
\]
Thus $\mathcal V_i(c)$ occurs if some pair violates the population error limit for class $i$ and passes its check at critical value $c$. We first prove
\begin{equation}
 \Prob\{\mathcal V_i(c)\}\le\psi_{n_i,\alpha_i}(c),\qquad i=1,2.
 \label{eq:classwise-violation-event}
\end{equation}

Step 6 shows that these events can be described using countably many threshold pairs, so their probabilities are well defined.

\paragraph{Step 2: Construct a common Bernoulli reference sequence.}
Fix $i\in\{1,2\}$ and abbreviate $N=n_i$, $\alpha=\alpha_i$, and $F(t)=P_i\{s(X)\le t\}$. We may generate the calibration scores as $S_j=F^{-1}(U_j)$, where $U_1,\ldots,U_N$ are independent $\operatorname{Unif}(0,1)$ variables and $F^{-1}(a)=\inf\{t\in\mathbb R:F(t)\ge a\}$. The quantile identity gives, simultaneously for every threshold,
\begin{equation}
 S_j\le t\quad\Longleftrightarrow\quad U_j\le F(t),
 \qquad j=1,\ldots,N,
 \label{eq:calibration-uniform-representation}
\end{equation}
with $F(-\infty)=0$ and $F(+\infty)=1$. These generated scores have the same joint distribution as the observed scores, including when scores have positive probability at individual values.

To orient errors toward the lower uniform tail, set $V_j=U_j$ for class 1 and $V_j=1-U_j$ for class 2. In either case, the $V_j$ are independent uniform variables. Define
\begin{equation}
 D_j=\min\left\{\frac{V_j}{\alpha},
                  \frac{1-V_j}{1-\alpha}\right\},
 \qquad Z_j=\ind\{V_j\le\alpha\}.
 \label{eq:reference-depth-label}
\end{equation}
The depth $D_j$ is the first value of $d$ at which $V_j$ enters $[0,\alpha d]\cup[1-(1-\alpha)d,1]$. As $d$ increases, these two tails expand from the endpoints. Ordering observations by their entry times gives one reference sequence for all candidate threshold pairs. Indeed, for $0\le d\le1$,
\[
 \Prob(D_j\le d,Z_j=1)=\alpha d,
 \qquad
 \Prob(D_j\le d,Z_j=0)=(1-\alpha)d.
\]
Thus $D_j$ is uniform on $(0,1)$ and independent of $Z_j\sim\operatorname{Bernoulli}(\alpha)$. The pairs $(D_j,Z_j)$ are independent across $j$.

Order the observations by increasing depth, and denote the corresponding labels by $Z_{(1)},\ldots,Z_{(N)}$. The depths are distinct with probability one. Conditional on the entire depth vector, this ordering is a fixed permutation of independent Bernoulli$(\alpha)$ labels. Hence the ordered labels are also independent Bernoulli$(\alpha)$ variables, conditionally on the depths. Write
\begin{equation}
 W_0=0,\qquad W_m=\sum_{j=1}^m Z_{(j)},
 \qquad m=1,\ldots,N.
 \label{eq:reference-error-count}
\end{equation}
The same sequence will be used for all threshold pairs. Its labels $Z_j$ are auxiliary; they need not equal a candidate classifier's error indicators. The construction also leaves the algorithm's treatment of tied scores unchanged.

\paragraph{Step 3: Compare every violating pair with the reference sequence.}
We first compare error counts at a common decision count. On the probability-one event of distinct depths, every pair with $R_i(h_{\ell,u})>\alpha$ satisfies
\begin{equation}
 W_{m_i(\ell,u)}\le e_i(\ell,u).
 \label{eq:simultaneous-count-domination}
\end{equation}
To prove this claim, fix a violating pair and let $w$ and $v$ be its population probabilities of an incorrect and a correct decision within class $i$:
\[
 (w,v)=
 \begin{cases}
  (F(\ell),\,1-F(u)),&i=1,\\
  (1-F(u),\,F(\ell)),&i=2.
 \end{cases}
\]
Since $w/(w+v)>\alpha$ and $w+v\le1$, we have
\begin{equation}
 d:=\frac{v}{1-\alpha}<1,
 \qquad \frac{w}{\alpha}>d,
 \qquad 1-v>\alpha.
 \label{eq:violating-tail-inequalities}
\end{equation}

For class 1, a correct decision has $V_j>1-v$, so that $Z_j=0$ and $D_j=(1-V_j)/(1-\alpha)<d$. An indecision has $w<V_j\le1-v$. Its two depth coordinates satisfy $V_j/\alpha>w/\alpha>d$ and $(1-V_j)/(1-\alpha)\ge d$, giving $D_j\ge d$. For class 2, reflection gives $V_j\ge1-v$ for a correct decision and $w\le V_j<1-v$ for an indecision. Consequently, correct decisions have $Z_j=0$ and $D_j\le d$, whereas indecisions have $D_j>d$. These weak and strict inequalities follow directly from the convention $s\le\ell$ and $s>u$ in \eqref{eq:candidate-rule}; no continuity of $F$ is required.

Thus every correct decision precedes every indecision in the depth ordering, in either class. This ordering is what links the candidate counts to the reference counts. Write $m=m_i(\ell,u)$, $e=e_i(\ell,u)$, and $k=k_i(\ell,u)=m-e$. All $k$ correct decisions lie among the first $m$ depths: otherwise, a correct decision after position $m$ would be preceded by at least $m$ decisions, contradicting the total of $m$. Each correct decision contributes a zero to the reference sequence, so $W_m\le m-k=e$. The same inequality is immediate when $k=0$. Once the uniforms are fixed, this comparison holds for every violating pair at once, including a pair chosen from the calibration data.

For fixed $m$, replacing a correct decision by an error multiplies the integrand defining $B_\alpha(e,m-e)$ by
\[
 \frac{r(1-\alpha)}{\alpha(1-r)}<1,
 \qquad 0<r<\alpha.
\]
Hence $B_\alpha(e,m-e)$ is decreasing in $e$. Combining this monotonicity with \eqref{eq:simultaneous-count-domination} gives
\[
 B_\alpha(e,m-e)\le B_\alpha(W_m,m-W_m)
\]
for every violating pair. The right-hand side is evaluated along one reference sequence; only the decision count varies with the pair.

\paragraph{Step 4: Apply the exact crossing probability.}
The reference counts $W_0,\ldots,W_N$ have the law of partial sums of independent Bernoulli$(\alpha)$ variables. Consequently, Lemma~\ref{lem:fixed-pair-calibration} gives
\begin{equation}
 \Prob\!\left\{
  \max_{1\le m\le N}B_\alpha(W_m,m-W_m)>c
 \right\}=\psi_{N,\alpha}(c).
 \label{eq:reference-crossing}
\end{equation}
If a violating pair passes its classwise check, Step 3 implies
\[
 B_\alpha(W_m,m-W_m)\ge B_\alpha(e,m-e)>c,
 \qquad m=m_i(\ell,u),\quad e=e_i(\ell,u).
\]
Here $m\ge1$, since $B_\alpha(0,0)=1\le c$. Thus, restoring the class index,
\begin{equation}
 \mathcal V_i(c)\subseteq
 \left\{
  \max_{1\le m\le n_i}B_{\alpha_i}(W_m,m-W_m)>c
 \right\}
 \label{eq:violation-crossing-inclusion}
\end{equation}
on one probability-one event. The decision count of a data-selected pair need not be a stopping time: it only selects one of the times already included in the crossing event. Equation~\eqref{eq:reference-crossing} proves \eqref{eq:classwise-violation-event}.

\paragraph{Step 5: Establish equality for a continuous score distribution.}
Suppose $F$ is continuous. We show that every reference crossing can be realized by a strictly violating pair with the same calibration counts. This step concerns all real threshold pairs and a single classwise check.

Suppose $B_\alpha(W_m,m-W_m)>c$ for some $m\ge1$, and let $d=D_{(m)}$. Almost surely $d<1$, and the first $m$ observations in depth order consist of the $W_m$ points with $V_j\le\alpha d$ and the $m-W_m$ points with $V_j\ge1-(1-\alpha)d$. We enlarge both tails without adding a calibration observation, increasing the erroneous tail enough to make the resulting rule's population error exceed $\alpha$. Choose $\epsilon>0$ sufficiently small that $d+2\epsilon<1$ and no $V_j$ lies in $(\alpha d,\alpha d+\epsilon]$. Next choose
\[
 0<\epsilon'<\min\{\epsilon,\epsilon(1-\alpha)/\alpha\}
\]
so that no $V_j$ lies in $[1-(1-\alpha)d-\epsilon',\,1-(1-\alpha)d)$. Such choices exist because the sample is finite. Set
\[
 w=\alpha d+\epsilon,\qquad
 v=(1-\alpha)d+\epsilon'.
\]
Then $w,v>0$, $w+v<1$, and $(1-\alpha)w-\alpha v=(1-\alpha)\epsilon-\alpha\epsilon'>0$.

Continuity of $F$ ensures that it attains every value in $(0,1)$, even if it has flat portions. We can therefore choose finite ordered thresholds with
\[
 (F(\ell),F(u))=
 \begin{cases}
  (w,1-v),&i=1,\\
  (v,1-w),&i=2.
 \end{cases}
\]
The two target values are ordered because $w+v<1$. For class 1, errors have $V_j\le w$ and correct decisions have $V_j>1-v$; for class 2, these comparisons are $V_j<w$ and $V_j\ge1-v$. Our choices of $\epsilon,\epsilon'$ preserve exactly the first $m$ reference observations under either convention. The pair thus makes $W_m$ errors and $m-W_m$ correct decisions, so it passes the classwise check. Its population error is $w/(w+v)>\alpha$. This proves the reverse inclusion in \eqref{eq:violation-crossing-inclusion}, and hence equality in \eqref{eq:classwise-violation-event} when $F_i$ is continuous.

\paragraph{Step 6: Handle all thresholds and combine the classwise bounds.}
Before combining the bounds, we show that the event $\mathcal V_i(c)$ can be described using countably many threshold pairs. Round every finite threshold $t$ up to $t^{(b)}=2^{-b}\lceil2^b t\rceil$, retaining infinite endpoints. The rounding map preserves $\ell\le u$, and its values decrease to the original thresholds as $b\to\infty$. For a finite calibration sample, all classifications eventually remain unchanged, including observations tied at a threshold. Let $w_i=P_i\{h_{\ell,u}(X)=3-i\}$ and $v_i=P_i\{h_{\ell,u}(X)=i\}$. A strict violation is equivalent to $(1-\alpha_i)w_i-\alpha_i v_i>0$. Since $F_i(t^{(b)})\to F_i(t)$, the strict violation persists for all sufficiently fine rounding. Thus $\mathcal V_i(c)$ equals the union of the corresponding events over ordered pairs of thresholds of the form $k/2^b$ ($k\in\mathbb Z$, $b\in\{1,2,\ldots\}$) and infinite endpoints, a countable collection.

For the critical values $c_i=c_{n_i,\alpha_i}(\eta_i)$, the defining crossing probabilities satisfy $\psi_{n_i,\alpha_i}(c_i)\le\eta_i$. Under the conditioning fixed at the start of the proof, $\mathcal V_i(c_i)$ depends only on the class-$i$ calibration scores. The two classwise samples are independent, so these events are independent under the same conditioning. Therefore
\[
 \begin{aligned}
 \Prob\{\mathcal V_1(c_1)^c\cap\mathcal V_2(c_2)^c\}
 &=\prod_{i=1}^2\bigl[1-\Prob\{\mathcal V_i(c_i)\}\bigr]\\
 &\ge\prod_{i=1}^2\{1-\psi_{n_i,\alpha_i}(c_i)\}\\
 &\ge(1-\eta_1)(1-\eta_2)\ge1-\delta.
 \end{aligned}
\]
This product argument concerns the two classwise events, each of which already ranges over all threshold pairs. It does not require independence of the errors of the data-selected classifier. On this intersection, every pair in the passing set $\mathcal A$ satisfies both population error limits. Because this statement covers all thresholds, it applies to the pair selected from the random grid $\mathcal G\times\mathcal G$ in \eqref{eq:selected-rule}. The finite grid search and the extrema defining the selected pair are measurable. If the grid contains no passing pair, the procedure returns $\widehat h\equiv0$, whose two selective errors are zero by convention. We have proved the claimed conditional guarantee for $\widehat h$. When calibration labels are random, average over them; if either class is absent, the all-indecision fallback satisfies both constraints. This gives the joint bound conditional on the training information alone. Averaging over that information proves \eqref{eq:joint-control}.

\begingroup
\paragraph{Label shift.}
To verify the consequence stated in the main text, consider any target class proportions $\pi_i>0$ with $\pi_1+\pi_2=1$, retaining the conditional distributions $P_i$. Write $\Prob_{\boldsymbol\pi}$ for this target distribution. For every rule $h$,
\[
 \Prob_{\boldsymbol\pi}\{h(X)\ne Y\mid Y=i,\ h(X)\ne0\}
 =\frac{\pi_i P_i\{h(X)=3-i\}}
        {\pi_i P_i\{h(X)\ne0\}}
 =R_i(h),
\]
with the convention $0/0=0$. Thus the validity event established above is the same for every $\boldsymbol\pi$. Conditioning on calibration labels removes any dependence of the bound on calibration class proportions, and conditioning on training allows its class proportions to differ as well. Finally, the decision-region inclusion in Proposition~\ref{prop:largest-region} gives $\sum_i\pi_i q_i(\widehat h)\ge\sum_i\pi_i q_i(h)$ for every passing grid rule $h$, proving the coverage claim for every target mixture.
\endgroup
\end{proof}

\subsection{Proof of Theorem~\ref{thm:coverage-oracle}}
\label{app:coverage-proof}

\begin{proof}
We construct a nearly optimal rule, show that it passes calibration, and transfer its coverage to the selected rule. Joint error control then completes the argument. All probabilities are conditional on the training information and calibration labels.
\paragraph{Step 1: Construct a nearly optimal rule with an error margin.}
Local regularity controls the coverage cost of an error margin. For positive decision probabilities,
\[
 \begin{gathered}
 q_i(h_{\ell,u})=F_i(\ell)+1-F_i(u),\quad i=1,2,\\
 R_1(h_{\ell,u})=\frac{F_1(\ell)}{q_1(h_{\ell,u})},\qquad
 R_2(h_{\ell,u})=\frac{1-F_2(u)}{q_2(h_{\ell,u})}.
 \end{gathered}
\]
Smoothness and positive $q_i(h_s^\star)$ make these quantities continuously differentiable near the ideal pair. Write $g_a=h_{\ell_s^\star+a v_\ell,\,u_s^\star+a v_u}$ for the threshold movement in \eqref{eq:local-oracle-regularity}. If at least one constraint is active, choose $b_{\mathrm{dir}}>0$ smaller than the magnitude of every active directional derivative. Differentiability then gives, for all sufficiently small $a\ge0$, $R_i(g_a)\le\alpha_i-b_{\mathrm{dir}}a$ for each active constraint. For inactive constraints, continuity preserves at least half the initial margin locally. Reducing $a_0>0$ if necessary ensures $R_i(g_a)\le\alpha_i-b_{\mathrm{dir}}a$ for both classes whenever $0\le a\le a_0$. Set $h_t=g_{t/b_{\mathrm{dir}}}$, so that the margin is measured directly by $t$. If both constraints are strict, instead take $h_t=h_s^\star$ for all sufficiently small $t\ge0$.

Set $\kappa=\tfrac12\min_i q_i(h_s^\star)>0$. By continuity of the decision probabilities and differentiability of coverage, constants $D_C,t_0>0$ depending only on the population problem satisfy, for $0\le t\le t_0$,
\begin{equation}
 R_i(h_t)\le\alpha_i-t,\qquad q_i(h_t)\ge\kappa,\quad i=1,2,
 \qquad
 0\le C_s^\star-C(h_t)\le D_Ct.
 \label{eq:local-comparison-path}
\end{equation}
The coverage difference is nonnegative because $h_t$ is feasible. Reduce $t_0$ so that $t_0<\min(\alpha_1,\alpha_2)$.

\paragraph{Step 2: Convert an empirical error margin into calibration evidence.}
We first lower-bound the calibration statistic using an empirical error margin. For one class, let $e,k$ be nonnegative integer counts with $m=e+k\ge1$, and put $\widehat\theta=e/m$. Suppose $\widehat\theta\le\alpha-\Delta$ for some $0<\Delta<\alpha$. For the likelihood ratio $\mathscr L_{\alpha,r}$ defined in \eqref{eq:calibration-lr}, differentiation of its logarithm gives $\frac{d}{dr}\{m^{-1}\log\mathscr L_{\alpha,r}(e,k)\} =(\widehat\theta-r)/\{r(1-r)\}$. Since its logarithm vanishes at $r=\alpha$, integration gives, for $r\in[\alpha-\Delta,\alpha-\Delta/2]$ with $\widehat\theta\le r<\alpha$,
\[
 \begin{aligned}
 \frac1m\log\mathscr L_{\alpha,r}(e,k)
 &=
 \int_r^\alpha\frac{z-\widehat\theta}{z(1-z)}\,dz\\
 &\ge 4\int_r^\alpha(z-r)\,dz
 =2(\alpha-r)^2
 \ge\frac{\Delta^2}{2}.
 \end{aligned}
\]
The first inequality uses $\widehat\theta\le r$ and $z(1-z)\le1/4$. Integrating over an interval of length $\Delta/2$ in \eqref{eq:mixture-evidence} gives
\begin{equation}
 B_\alpha(e,k)
 \ge\frac{\Delta}{2\alpha}
 \exp\left(\frac{m\Delta^2}{2}\right).
 \label{eq:coverage-evidence-bound}
\end{equation}

\paragraph{Step 3: Choose the margin and show that both checks pass.}
Choose the margin to make the evidence bound exceed both critical values. Define
\[
 \zeta=\log\frac{n_{\min}c_{\max}}{\beta},\qquad
 a_{\mathrm{cal}}=\sqrt{\frac{128}{\kappa}},\qquad
 t=a_{\mathrm{cal}}\sqrt{\frac{\zeta}{n_{\min}}}.
\]
With $c_0=(t_0/a_{\mathrm{cal}})^2$, the sample-size condition ensures $t\le t_0$. Step 1 therefore applies to $h_t$, which is independent of the calibration observations. Since $n_{\min}\ge2$ and $c_{\max}\ge1$, we have $\zeta\ge\log2$ and $t\ge n_{\min}^{-1/2}$. Using $\alpha_i<1$ and $c_i\le c_{\max}$, we obtain
\[
 \log\frac{4\alpha_i c_i}{t}
 \le \log4+\tfrac12\log n_{\min}+\log c_{\max}
 \le 3\zeta.
\]
Here $\log4\le2\log n_{\min}$ because $n_{\min}\ge2$. Also, $\log8\le3\zeta$ and $\log(1/\beta)\le\zeta$. Together with the definition of $t$, these inequalities give
\begin{equation}
 \log\frac{4\alpha_i c_i}{t}\le3\zeta,\qquad
 \log\frac8\beta\le4\zeta,\qquad
 n_i\kappa t^2\ge128\zeta,\quad i=1,2.
 \label{eq:coverage-calibration-size}
\end{equation}

We next show that the empirical errors retain half the population margin. For the thresholds $(\ell_t,u_t)$ of $h_t$, write $\widetilde m_i=m_i(\ell_t,u_t)$ and $\widetilde e_i=e_i(\ell_t,u_t)$, using \eqref{eq:calibration-counts}. Because $h_t$ is fixed, its decision indicators are independent Bernoulli variables, and errors conditional on a decision have probability $R_i(h_t)$. Thus
\[
 \widetilde m_i\sim\operatorname{Bin}\{n_i,q_i(h_t)\},
 \qquad
 \widetilde e_i\mid \widetilde m_i=m\sim\operatorname{Bin}\{m,R_i(h_t)\}.
\]
Using $q_i(h_t)\ge\kappa$ and $R_i(h_t)\le\alpha_i-t$, consider the event
\[
 \mathcal E_{i,t}=
 \left\{
 \widetilde m_i\ge\frac{n_i\kappa}{2},\quad
 \widetilde e_i\le \widetilde m_i\left(\alpha_i-\frac{t}{2}\right)
 \right\}.
\]
A binomial Chernoff bound controls the decision count: $\Prob\{\widetilde m_i<n_i\kappa/2\}\le\exp(-n_i\kappa/8)$. Conditional on $\widetilde m_i=m\ge n_i\kappa/2$, Hoeffding's inequality \citep{hoeffding1963probability} bounds the probability that $\widetilde e_i/m$ exceeds $\alpha_i-t/2$ by $\exp(-mt^2/2)$. Consequently,
\begin{equation}
 \Prob\{\mathcal E_{i,t}^{c}\}
 \le
 \exp\left(-\frac{n_i\kappa}{8}\right)
 +
 \exp\left(-\frac{n_i\kappa t^2}{4}\right).
 \label{eq:coverage-certification-failure}
\end{equation}
On $\mathcal E_{i,t}$, the evidence bound \eqref{eq:coverage-evidence-bound} with $\Delta=t/2$ gives
\[
 B_{\alpha_i}(\widetilde e_i,\widetilde m_i-\widetilde e_i)
 \ge
 \frac{t}{4\alpha_i}
 \exp\left(\frac{n_i\kappa t^2}{16}\right)
 >c_i.
\]
Indeed, the exponent is at least $8\zeta$, whereas $\log\{4\alpha_i c_i/t\}\le3\zeta$ by \eqref{eq:coverage-calibration-size}. For the failure probability, $t<1$ and the same sample-size bound give
\[
 \Prob\{\mathcal E_{i,t}^{c}\}
 \le e^{-16\zeta}+e^{-32\zeta}
 \le 2e^{-4\zeta}\le\beta/4.
\]
The last inequality uses $\log(8/\beta)\le4\zeta$. A union bound shows that $h_t$ passes both calibration checks with probability at least $1-\beta/2$.

\begingroup
\paragraph{Step 4: Transfer coverage through the calibration grid.}
The comparison rule need not have thresholds on the search grid. We first bound the coverage lost when replacing it by a grid rule with the same calibration decisions. Let $\widehat F_i(z)=n_i^{-1}\sum_{j=1}^{n_i} \ind\{s(X_{ij})\le z\}$ and set $\varepsilon_i=\{\log(8/\beta)/(2n_i)\}^{1/2}$. The Dvoretzky--Kiefer--Wolfowitz inequality \citep{massart1990tight} and a union bound give, with probability at least $1-\beta/2$,
\[
 \sup_z|\widehat F_i(z)-F_i(z)|\le\varepsilon_i,\qquad i=1,2.
\]
For each threshold $z$, write $z^{\mathcal G}=\max\{g\in\mathcal G:g\le z\}$. This replacement preserves comparisons with every calibration score and the ordering of the two thresholds, including ties. In particular, $\widehat F_i(z^{\mathcal G})=\widehat F_i(z)$. On the event above, therefore,
\[
 |F_i(z^{\mathcal G})-F_i(z)|\le2\varepsilon_i.
\]
Let $h_t^{\mathcal G}=h_{\ell_t^{\mathcal G},u_t^{\mathcal G}}$. It has the same classwise error and correct-decision counts as $h_t$, so it passes both checks whenever $h_t$ does. Using $q_i(h_{\ell,u})=F_i(\ell)+1-F_i(u)$ gives
\begin{equation}
 |q_i(h_t^{\mathcal G})-q_i(h_t)|\le4\varepsilon_i,
 \qquad i=1,2.
 \label{eq:uniform-coverage-bound}
\end{equation}
On the intersection of this event and the event in Step 3 that $h_t$ passes both checks, Proposition~\ref{prop:largest-region} gives
\[
 q_i(\widehat h)\ge q_i(h_t^{\mathcal G})
 \ge q_i(h_t)-4\varepsilon_i,\qquad i=1,2.
\]
This intersection has probability at least $1-\beta$. Multiplying by the target-population proportions and summing yields $C(\widehat h)\ge C(h_t)-4\sum_i\pi_i\varepsilon_i$. Together with \eqref{eq:local-comparison-path}, this implies
\[
 C_s^\star-C(\widehat h)
 \le D_Ct+4\sum_i\pi_i\varepsilon_i
 \le(D_Ca_{\mathrm{cal}}+4\sqrt2)
       \sqrt{\frac{\zeta}{n_{\min}}}.
\]
The last inequality uses $n_i\ge n_{\min}$, $\pi_1+\pi_2=1$, and $\log(8/\beta)\le4\zeta$. Thus the coverage bound holds with probability at least $1-\beta$, with $K=D_Ca_{\mathrm{cal}}+4\sqrt2$.
\endgroup

\paragraph{Step 5: Combine the coverage bound with joint error control.}
Theorem~\ref{thm:simultaneous} guarantees both population error constraints for $\widehat h$ with probability at least $1-\delta$. On this event, $\widehat h$ belongs to the family in \eqref{eq:score-coverage-oracle}, including when it is the all-indecision rule, so $C(\widehat h)\le C_s^\star$. Intersecting the error-control and coverage events proves \eqref{eq:coverage-oracle-bound} with probability at least $1-\delta-\beta$.
\end{proof}

\section{Additional results}
\label{sec:extensions}

\subsection{Infinite exchangeability}
\label{app:exchangeability}

An infinite sequence of labeled observations $(X_j,Y_j)_{j\ge1}$ is \emph{exchangeable} if reordering any finite initial segment leaves its joint distribution unchanged. That is, for every $m\ge1$ and every permutation $(j_1,\ldots,j_m)$ of $(1,\ldots,m)$,
\[
 \bigl((X_1,Y_1),\ldots,(X_m,Y_m)\bigr)
 \stackrel{d}{=}
 \bigl((X_{j_1},Y_{j_1}),\ldots,(X_{j_m},Y_{j_m})\bigr).
\]
The requirement applies to a whole infinite sequence, even though the observed sample is finite.

By de Finetti's theorem, on a standard Borel sample space, such a sequence is independent and identically distributed conditional on a random population distribution $Q$ \citep{foygelbarber2023definetti}. The observations may therefore be dependent through the shared $Q$. Ordinary i.i.d. sampling is recovered when $Q$ is nonrandom.

\begin{corollary}
\label{cor:infinite-exchangeability}
Suppose $(X_j,Y_j)_{j\ge1}$ is infinitely exchangeable on the standard Borel space $\mathcal X\times\{1,2\}$. Let $Q$ be its random population distribution, with $Q(Y=i)>0$ almost surely, and write $Q_i=Q(X\in\cdot\mid Y=i)$. Fit $s$ using a fixed initial training block and, if needed, randomization independent of the sequence. For calibration, retain the first $n_i$ subsequent observations of class $i$, where $n_1,n_2\ge1$ are fixed in advance. Define
\[
 R_i(h;Q_i)=
 \frac{Q_i\{h(X)=3-i\}}{Q_i\{h(X)\ne0\}},
 \qquad i=1,2,
\]
with $0/0=0$. Fix targets $\alpha_i\in(0,1)$ and allowances $\eta_i\in(0,1)$ such that $(1-\eta_1)(1-\eta_2)\ge1-\delta$ before calibration. Apply the NPI calibration and selection steps. Then
\begin{equation}
 \Prob\{R_1(\widehat h;Q_1)\le\alpha_1,\;
        R_2(\widehat h;Q_2)\le\alpha_2\}\ge1-\delta.
 \label{eq:exchangeable-control}
\end{equation}
The same guarantee holds conditional on the training information.
\end{corollary}

\begin{proof}
The key is to condition on the random population distribution $Q$. We verify that the retained calibration sample then has the classwise independent sampling law required by Theorem~\ref{thm:simultaneous}, apply that theorem, and average over $Q$.

\paragraph{Step 1: Identify the conditional calibration law.}
Let $\mathcal F_{\mathrm{train}}$ denote the information in the fixed initial training block and independent training seed. Conditional on $Q$, de Finetti's theorem makes the labeled observations independent and identically distributed from $Q$. Conditioning further on $\mathcal F_{\mathrm{train}}$ fixes the fitted score and leaves the post-training sequence independent and identically distributed from $Q$.

To obtain the conditional law of the retained calibration features, condition additionally on all post-training labels. The retained indices are now fixed, and their features are mutually independent; a feature whose label is $i$ has law $Q_i$. Both class quotas are reached almost surely because $Q(Y=i)>0$. The retained features are therefore $n_1$ independent draws from $Q_1$ and $n_2$ independent draws from $Q_2$, with the two groups independent. This law does not depend on the post-training label sequence, so it remains the conditional law given $(Q,\mathcal F_{\mathrm{train}})$ after those labels are integrated out. Thus retaining observations until both class quotas are filled does not change this conditional feature law.

\paragraph{Step 2: Apply the joint error bound conditionally.}
The assumptions of Theorem~\ref{thm:simultaneous} thus hold conditionally, with $P_i$ replaced by $Q_i$. The critical values $c_i=c_{n_i,\alpha_i}(\eta_i)$ and their reference crossing probabilities depend only on the fixed class sample sizes, error targets, and allowances; they do not depend on $Q$. In particular, almost surely,
\[
 \begin{aligned}
 \Prob\!\left\{
  \exists(\ell,u)\in\mathcal A:
  R_i(h_{\ell,u};Q_i)>\alpha_i\text{ for some }i
  \mid Q,\mathcal F_{\mathrm{train}}
 \right\}
 &\le1-\prod_{i=1}^2\{1-\psi_{n_i,\alpha_i}(c_i)\}\\
 &\le1-(1-\eta_1)(1-\eta_2)\le\delta.
 \end{aligned}
\]
On the complementary event, every passing pair satisfies both error requirements, including the selected pair. If no pair passes, $\widehat h\equiv0$ has both selective errors zero.

\paragraph{Step 3: Average over the random population distribution.}
Writing $\mathcal V_{\mathrm{ex}}=\{R_i(\widehat h;Q_i)>\alpha_i\text{ for some }i\}$, we obtain
\[
 \Prob(\mathcal V_{\mathrm{ex}}\mid\mathcal F_{\mathrm{train}})
 =\E\!\left\{
    \Prob(\mathcal V_{\mathrm{ex}}\mid Q,\mathcal F_{\mathrm{train}})
    \mid\mathcal F_{\mathrm{train}}
   \right\}
 \le\delta
 \qquad\text{almost surely}.
\]
Averaging over the training information proves \eqref{eq:exchangeable-control}.
\end{proof}

The selective errors in \eqref{eq:exchangeable-control} concern the population $Q$ shared by calibration and future observations, rather than the marginal population obtained by averaging over $Q$. The infinite-sequence assumption is needed for the conditional independence used in the proof; exchangeability of a finite sample alone does not suffice.

\subsection{Calibration sample requirements}
\label{subsec:calibration-requirements}
\label{app:calibration-properties}

Even a score that separates the two classes needs enough calibration observations to pass the error checks. For a fixed number of decisions $m$, the statistic $B_\alpha$ is largest when none of these decisions is incorrect, in which case
\begin{equation}
 B_\alpha(0,m)
 =\frac{(1-\alpha)^{-m}-(1-\alpha)}{\alpha(m+1)}.
 \label{eq:zero-error-evidence}
\end{equation}
This most favorable case determines the minimum number of decisions required by each classwise check.

\begin{proposition}
\label{prop:minimum-calibration}
For $n\ge1$ and $\alpha,\eta\in(0,1)$, define
\[
 m_0(n,\alpha,\eta)
 =\min\left\{
 m\in\{0,\ldots,n\}:B_\alpha(0,m)>c_{n,\alpha}(\eta)
 \right\},
\]
with $m_0=+\infty$ when the set is empty. A classwise check can pass only if at least $m_0(n_i,\alpha_i,\eta_i)$ calibration observations from class $i$ receive decisions. The check admits a passing count configuration if and only if
\[
 (1-\alpha_i)^{n_i}\le\eta_i,
 \quad\text{equivalently}\quad
 n_i\ge\left\lceil
 \frac{\log\eta_i}{\log(1-\alpha_i)}\right\rceil.
\]
If this condition fails for either class, NPI returns indecisions for all observations. If it holds for both classes and every class 1 calibration score exceeds every class 2 calibration score, a pair passes both checks and classifies every calibration observation correctly.
\end{proposition}

\begin{proof}
For fixed $m$, replacing a correct decision by an error multiplies the integrand defining $B_\alpha(e,m-e)$ by $r(1-\alpha)/\{\alpha(1-r)\}<1$. Thus $B_\alpha(e,m-e)\le B_\alpha(0,m)$. Setting $e=0$ and integrating gives \eqref{eq:zero-error-evidence}. Adding a correct decision multiplies the integrand by $(1-r)/(1-\alpha)>1$, so $B_\alpha(0,m)$ strictly increases with $m$. No count with $m<m_0(n,\alpha,\eta)$ can therefore pass the exact check.

To determine when any count can pass, consider the reference maximum $\max_{1\le m\le n}T_m$. Its largest possible value is $B_\alpha(0,n)$, attained exactly when all $n$ reference Bernoulli variables are zero. Indeed, every state with an error has a smaller statistic than the zero-error state at the same time, and every earlier zero-error state has a smaller statistic than $B_\alpha(0,n)$. Hence the probability of this largest value is $(1-\alpha)^n$.

If any count can pass, then $B_\alpha(0,n)>c_{n,\alpha}(\eta)$. The all-zero path must therefore be included in the crossing event, giving $(1-\alpha)^n\le\psi_{n,\alpha}\{c_{n,\alpha}(\eta)\}\le\eta$. Conversely, suppose $(1-\alpha)^n\le\eta$. The reference maximum takes finitely many values and $B_\alpha(0,n)>1$. We may choose $c\ge1$ below this largest value and at least as large as every other possible value. Then $\psi_{n,\alpha}(c)=(1-\alpha)^n\le\eta$, so $c_{n,\alpha}(\eta)\le c<B_\alpha(0,n)$. The count configuration with $n$ decisions and no errors therefore passes. This proves the equivalence and its logarithmic form.

Finally, when every class 1 calibration score exceeds every class 2 calibration score, set both thresholds equal to the largest class 2 score. This grid rule makes $n_i$ correct decisions and no errors in each class. The equivalence just proved shows that it passes both checks under the stated sample-size conditions.
\end{proof}

The sample-size condition ensures that a classwise check can pass in the most favorable configuration. The actual minimum decision count $m_0(n,\alpha,\eta)$ depends on the horizon $n$, and passing for observed data also depends on the score separation and error counts.

\subsection{The role of local regularity}
\label{app:local-regularity-example}

The example in Remark~\ref{rem:local-oracle-regularity} shows that smooth score distributions and positive ideal coverage $C_s^\star$ alone do not ensure the coverage bound in Theorem~\ref{thm:coverage-oracle}.

\begin{proposition}
\label{prop:local-regularity-example}
Let $P_1=P_2=\operatorname{Unif}(0,1)$, take $s(x)=x$, and set $\alpha_1=\alpha_2=1/2$. Then $C_s^\star=1$, but the rule returned by Algorithm~\ref{alg:calibration} satisfies
\[
 \Prob\{C_s^\star-C(\widehat h)=1\}\ge1-\delta
 \qquad\text{for every finite }n_1,n_2\ge1.
\]
\end{proposition}

\begin{proof}
We first identify the population optimum and then show why the calibration grid almost surely contains no feasible rule with positive coverage.

The common threshold $1/2$ classifies every observation and meets both error limits, so $C_s^\star=1$. Both score distribution functions are continuously differentiable near this ideal pair, and $q_i(h_s^\star)=1$. However, every rule with positive coverage satisfies $R_1(h)+R_2(h)=1$. No movement of the thresholds can therefore reduce both active error rates, and no such rule satisfies both constraints strictly.

We now characterize the feasible pairs and compare them with the calibration grid. For a finite calibration-grid pair $0<\ell\le u<1$,
\[
 R_1(h_{\ell,u})=\frac{\ell}{\ell+1-u},
 \qquad
 R_2(h_{\ell,u})=\frac{1-u}{\ell+1-u}.
\]
Both errors are at most $1/2$ exactly when $\ell+u=1$. With probability one, no two calibration scores sum to one and no calibration score equals $1/2$, since the finitely many calibration scores have a joint continuous distribution. A grid pair with an infinite endpoint either makes no decisions or assigns only one class, giving error one in the other class. Thus, almost surely, every grid rule with positive coverage violates at least one target.

It remains to connect this grid property to the output of NPI. On the simultaneous-validity event of Theorem~\ref{thm:simultaneous}, none of these rules can pass both calibration checks, and NPI returns only indecisions. This event has probability at least $1-\delta$, proving the claim.
\end{proof}

\subsection{Other calibration boundaries}
\label{app:general-boundaries}

The count comparison in the proof of Theorem~\ref{thm:simultaneous} also applies to calibration checks specified directly through their largest admissible error counts. The following result separates this comparison from the choice of calibration statistic.

\begin{proposition}
\label{prop:general-boundaries}
Under the sampling assumptions of Theorem~\ref{thm:simultaneous}, condition on the training information and calibration labels, with $n_i\ge1$, $i=1,2$. For each class, let $b_i(0)=-1$ and $b_i(m)\in\{-1,0,\ldots,m\}$, $m=1,\ldots,n_i$, be fixed without using the calibration features. Let $W_{i,m}$ be the partial sums of independent Bernoulli$(\alpha_i)$ variables, and define
\[
 \Psi_i(b_i)
 =\Prob\{W_{i,m}\le b_i(m)\text{ for some }1\le m\le n_i\}.
\]
Then
\begin{equation}
 \Prob\!\left\{
 \begin{array}{l}
 \exists\,\ell\le u:\ R_i(h_{\ell,u})>\alpha_i,\\
 e_i(\ell,u)\le b_i\{m_i(\ell,u)\}
 \end{array}\right\}\le\Psi_i(b_i),
 \label{eq:general-boundary-control}
\end{equation}
with equality when the class-$i$ score distribution is continuous. The crossing probability is computed by
\[
 \begin{aligned}
 p_{i,m}(w)
 &=\{\alpha_i p_{i,m-1}(w-1)+(1-\alpha_i)p_{i,m-1}(w)\}
       \ind\{w>b_i(m)\},\\
 \Psi_i(b_i)&=1-\sum_{w=0}^{n_i}p_{i,n_i}(w),
 \end{aligned}
\]
where $p_{i,0}(0)=1$ and entries outside $0\le w\le m$ are zero. If $\Psi_i(b_i)\le\eta_i$ and $(1-\eta_1)(1-\eta_2)\ge1-\delta$, then, with conditional probability at least $1-\delta$, every pair satisfying both checks $e_i\le b_i(m_i)$ obeys both population error limits. Any measurable selection of such a pair, with the all-indecision fallback when no pair passes, therefore satisfies the joint guarantee.
\end{proposition}

\begin{proof}
For class $i$, use the uniform variables and depth ordering constructed in Steps 2--3 of the proof of Theorem~\ref{thm:simultaneous}. On one probability-one event, the resulting Bernoulli reference counts satisfy
\[
 W_{i,m_i(\ell,u)}\le e_i(\ell,u)
 \quad\text{for every pair with }R_i(h_{\ell,u})>\alpha_i.
\]
If such a pair passes, then $W_{i,m_i}\le e_i\le b_i(m_i)$, and $m_i\ge1$ because $b_i(0)=-1$. Its passing event is thus contained in the reference crossing event, proving \eqref{eq:general-boundary-control}. As in Step 6 of that proof, rounding thresholds upwards to successively finer dyadic grids eventually preserves all calibration counts and any strict population violation. The event can therefore be expressed using countably many pairs and has a well-defined probability.

When the score distribution is continuous, Step 5 of the same proof constructs, for each reference state $(m,W_{i,m})$, a strictly violating real threshold pair with exactly $m$ decisions and $W_{i,m}$ errors. If $W_{i,m}\le b_i(m)$, this pair passes. This proves the reverse inclusion and equality.

For the recursion, $p_{i,m}(w)$ is the probability of reaching $w$ at step $m$ without crossing. A path arrives from $w-1$ with probability $\alpha_i$ or from $w$ with probability $1-\alpha_i$, and survives exactly when $w>b_i(m)$. Summing the surviving probabilities gives $1-\Psi_i(b_i)$. Finally, the two classwise violation events depend on independent calibration samples under the stated conditioning. Multiplying their complementary probabilities gives at least $(1-\eta_1)(1-\eta_2)\ge1-\delta$. On this event, every passing pair satisfies both constraints; the fallback has both errors zero.
\end{proof}

To combine passing pairs by the rule in \eqref{eq:selected-rule}, a sufficient additional condition is
\begin{equation}
 b_i(m)\le b_i(m+1)\le b_i(m)+1,
 \qquad 0\le m<n_i.
 \label{eq:boundary-shape}
\end{equation}
The first inequality preserves a passing check when a correct decision is added. The second preserves it when an error is removed: $e\le b_i(m)$ implies $e-1\le b_i(m-1)$. The count comparisons in Proposition~\ref{prop:largest-region} then show that the combined pair also passes and has the largest decision region among passing grid pairs.

The boundaries $b_{\alpha,c}$ used in this paper satisfy \eqref{eq:boundary-shape}. Adding a correct decision increases $B_\alpha$, giving $b_{\alpha,c}(m+1)\ge b_{\alpha,c}(m)$; this inequality is immediate if the latter value is $-1$. If $e=b_{\alpha,c}(m+1)\ge1$, removing an error increases $B_\alpha$, so $e-1\le b_{\alpha,c}(m)$. If $e\le0$, the same upper bound follows from $b_{\alpha,c}(m)\ge-1$. Thus the default calibration retains both the martingale interpretation and the largest-region property. Other boundaries can be calibrated by the recursion above; the coverage rate in Theorem~\ref{thm:coverage-oracle} additionally uses the particular lower bound on $B_\alpha$ proved in Appendix Section~\ref{app:coverage-proof}.

\section{Implementation details}
\label{app:calibration-tables}

\subsection{LTT implementations}
\label{app:ltt-implementation}

We apply LTT to a finite collection of threshold pairs $\{(\ell_j,u_j):j=1,\ldots,K_{\mathrm{LTT}}\}$, each defining a rule of the form \eqref{eq:candidate-rule}. For both the simulations and real-data analyses, the threshold grid consists of the training-score quantiles at probabilities $0,0.01,\ldots,1$, together with $-\infty$ and $+\infty$. We consider all pairs with $\ell_j\le u_j$ and retain those with positive decision counts in both classes and training selective errors no larger than their targets. The candidate collection is thus fixed before calibration.

For each pair, we test the null hypothesis that at least one population error exceeds its target,
\[
 H_j:\ R_1(h_{\ell_j,u_j})>\alpha_1
       \quad\text{or}\quad R_2(h_{\ell_j,u_j})>\alpha_2.
\]
Let $M_{ij}$ and $E_{ij}$ denote the numbers of decisions and errors, respectively, among calibration observations from actual class $i$. The classwise binomial $p$-values and their combined $p$-value are
\[
 p_{ij}
 =
 \Prob\!\left\{
 \operatorname{Bin}(M_{ij},\alpha_i)\le E_{ij}
 \right\},
 \qquad
 p_j=\max\{p_{1j},p_{2j}\},
\]
setting $p_{ij}=1$ when $M_{ij}=0$. Conditional on $M_{ij}$ and the training data, the error count has a binomial distribution with success probability $R_i(h_{\ell_j,u_j})$. Taking the maximum therefore gives a valid $p$-value for $H_j$: rejection provides evidence that both error constraints hold \citep{angelopoulos2025learn}. We use two procedures to control the family-wise error rate over the candidate collection at level $\delta=0.05$.

\emph{LTT Holm} applies Holm's step-down adjustment \citep{holm1979simple}. It sorts the calibration $p$-values and rejects successively while $p_{(j)}\le\delta/(K_{\mathrm{LTT}}-j+1)$, stopping at the first nonrejection.

\emph{LTT sequence} fixes a testing order using training data. For each pair, we evaluate the same binomial-tail formula with the training counts,
\[
 \widetilde p_j=\max_{i=1,2}
 \Prob\!\left\{
 \operatorname{Bin}(M^{\mathrm{train}}_{ij},\alpha_i)
 \le E^{\mathrm{train}}_{ij}\right\}.
\]
We order candidates by increasing $\widetilde p_j$ and retain a candidate only if its training coverage exceeds that of all previously retained candidates. These training statistics determine the sequence only; the tests themselves use the calibration $p$-values $p_j$. We use a single starting point, $|\mathcal J|=1$ in Algorithm~1 of \citet{angelopoulos2025learn}, so each test compares $p_j$ with $\delta$ and testing stops at the first nonrejection. For either implementation, we select the rejected candidate with the largest calibration coverage, breaking ties by the smaller lower threshold and then the smaller upper threshold. We return indecisions everywhere if no candidate is rejected.

\begingroup
For the final LTT selection, calibration coverage is $\sum_{i=1}^2\widehat\pi_i M_{ij}/n_i$. In the simulations, we use $\widehat\pi_i=1/2$. In the real-data analyses, let $N_i$ denote the number of class $i$ observations in the training and calibration samples combined; we set $\widehat\pi_i=N_i/(N_1+N_2)$. The training coverage used to construct the LTT sequence uses the training class proportions instead. These weights are used only for LTT coverage comparisons.
\endgroup

\section{Credit-card default prediction}
\label{app:credit-application}
\label{subsec:credit-application}

We use the comparison and evaluation protocol in Section~\ref{sec:real_data}, with 1,000 random partitions. Our second application uses the Default of Credit Card Clients data set from the UCI Machine Learning Repository \citep{yeh2009creditdata}. The data contain 30,000 clients in Taiwan, with the response indicating default payment in the following month. Class 1 comprises 23,364 clients without recorded default, and class 2 comprises 6,636 with recorded default. The 23 predictors describe credit limits, demographics, and repayment, billing, and payment histories. We exclude the client identifier and retain all recorded category codes. Full indicator encoding of sex, education, and marital status yields 33 numeric inputs.

Each partition assigns 10,000 clients to training, 10,000 to calibration, and 10,000 to testing. The score is the equal-weight average of nondefault probabilities from a 300-tree random forest \citep{breiman2001randomforests} and an XGBoost model \citep{chen2016xgboost} with 300 rounds, depth 3, and learning rate $0.05$. Both learners use the same 33 inputs, with settings held fixed across partitions; further details appear in Section~\ref{app:credit-score}. The ensemble choice was informed by a preliminary, NPI-directed comparison of four learners at the $15\%$ target.

We evaluate $\alpha\in\{0.05,0.10,0.15,0.20,0.25,0.30\}$; Figure~\ref{fig:credit_alpha} displays the four largest targets. The two selective errors concern nondefaulting clients classified as defaulters and defaulters classified as nondefaulting clients, respectively, among definitive classifications within each actual class.

\begin{figure}[!t]
    \centering
    \resizebox{\textwidth}{!}{%
        \input{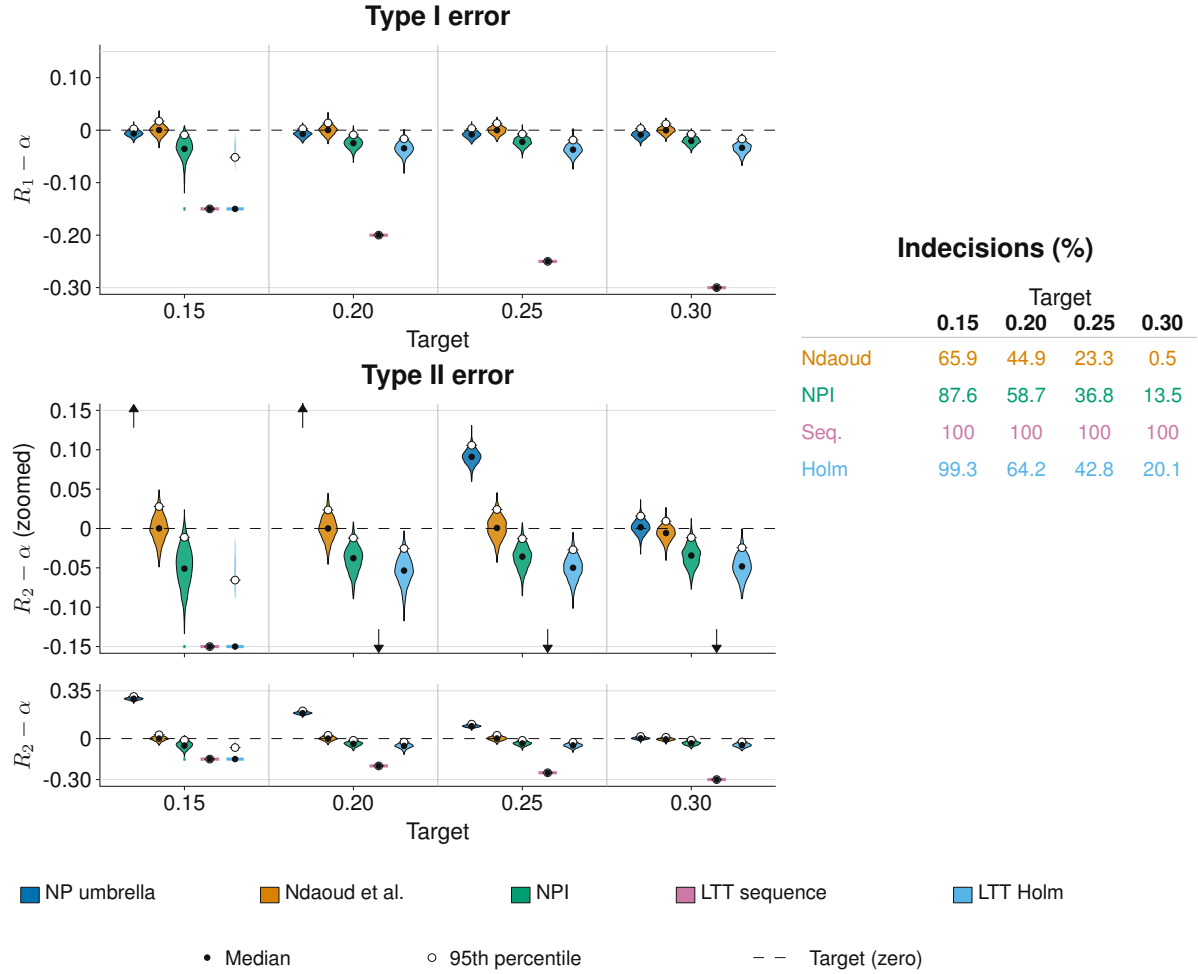}%
    }
    \caption{\small Credit-card default results over 1,000 random partitions, using an equal-weight random forest and boosted-tree ensemble with 10,000 training, 10,000 calibration, and 10,000 test observations. The panels show held-out selective error minus $\alpha$, using the colors and markers in Figure~\ref{fig:compas_alpha}. All 1,000 partitions enter the error distributions, quantiles, and mean indecision percentages, using $0/0=0$. Zero-error bars and positive-error violin widths are scaled by their respective proportions of partitions; arrows mark off-scale quantiles or zero-error bars.
    }
    \label{fig:credit_alpha}
\end{figure}

Figure~\ref{fig:credit_alpha} shows that NPI achieves lower mean test indecision than LTT Holm at each displayed target. At the $15\%$ target, these percentages are $87.6\%$ and $99.3\%$, respectively. At targets of $20\%$, $25\%$, and $30\%$, NPI's mean indecision decreases to $58.7\%$, $36.8\%$, and $13.5\%$, compared with $64.2\%$, $42.8\%$, and $20.1\%$ for LTT Holm. The 95th percentiles of NPI's two test errors remain below the target in all four settings: they are $14.1\%$ and $13.9\%$ at the $15\%$ target, and $29.2\%$ and $28.8\%$ at the $30\%$ target.

LTT sequence returns indecisions everywhere in all 1,000 partitions at every displayed target. Its first candidate fails calibration in each case, terminating the sequence before any rule passes. Training errors from these flexible learners can be overly optimistic when constructing the testing order. NPI does not require this ordering and can select among all passing pairs on its calibration grid.

The additional targets of $5\%$ and $10\%$ permit few decisions for any of the three procedures. NPI's mean indecision is $100\%$ and $99.94\%$, respectively; both LTT implementations return indecisions everywhere at these targets.

\subsection{Score fitting}
\label{app:credit-score}

The credit-card analysis uses 20 numeric predictors and full indicator encoding of sex (codes 1--2), education (0--6), and marital status (0--3), giving 33 inputs. These fixed encodings retain all recorded categories; no observations are removed or category codes merged. Within each partition, both learners are fitted only to the 10,000 training observations, with no hyperparameter tuning or early stopping.

We fit a probability forest using \texttt{ranger} \citep{wright2017ranger}, with 300 bootstrap trees, five candidate predictors per split, Gini splitting, and minimum node size 1. Each bootstrap sample has the same size as the training sample, and tree depth is unrestricted. The XGBoost model uses a binary logistic objective, 300 rounds, depth 3, learning rate $0.05$, and histogram tree construction with 256 bins. It uses all rows and predictors at each round, minimum child weight 1, minimum split loss 0, $L_2$ penalty 1, $L_1$ penalty 0, and initial probability $0.5$. The fitted score is the mean of the two estimated nondefault probabilities. We evaluate both fitted models directly on each sample; training scores use fitted predictions rather than out-of-bag predictions.

\end{document}

%% file: arxiv_figures/npi_threshold_cartoon.tex
% Mathematical cartoon, not a plot of simulated observations.
% Requires tikz, amsmath, helvet and T1 fontenc;
% TikZ libraries arrows.meta and decorations.pathreplacing.
\begingroup
% Approximately ggplot2 base_size=7; fonts remain local to this picture.
\def\npiBaseFont{\fontsize{7}{9}\selectfont}
\def\npiEmphasisFont{\fontsize{8}{10}\selectfont}
\fontfamily{phv}\selectfont
\begin{tikzpicture}[x=1.12cm,y=.80cm,>=Latex,
  font=\fontfamily{phv}\npiBaseFont]
\definecolor{npiink}{HTML}{263746}
\definecolor{npicorrect}{HTML}{008B73}
\definecolor{npierror}{HTML}{D67A22}
\definecolor{npimiddle}{HTML}{A8AFB5}
\definecolor{npbaseline}{HTML}{76528D}
\path[use as bounding box] (0,.24) rectangle (12.0,7.98);

% A shared score axis: the distribution is shown separately for each true class
% so that the numerator and denominator of each selective error are visible.
\fill[npimiddle!9] (5.19,1.90) rectangle (6.81,7.42);
\foreach \base/\mu/\idx in {5.00/1/1,1.90/-1/2}{
  \ifnum\idx=1
    \colorlet{lefttail}{npierror}
    \colorlet{righttail}{npicorrect}
  \else
    \colorlet{lefttail}{npicorrect}
    \colorlet{righttail}{npierror}
  \fi
  \fill[lefttail!40] (.6,\base)
    -- plot[domain=-4:-.6,samples=100]
      ({6+1.35*\x},{\base+2.15*exp(-(\x-(\mu))^2/2)})
    -- (5.19,\base) -- cycle;
  \fill[npimiddle!30] (5.19,\base)
    -- plot[domain=-.6:.6,samples=75]
      ({6+1.35*\x},{\base+2.15*exp(-(\x-(\mu))^2/2)})
    -- (6.81,\base) -- cycle;
  \fill[righttail!40] (6.81,\base)
    -- plot[domain=.6:4,samples=100]
      ({6+1.35*\x},{\base+2.15*exp(-(\x-(\mu))^2/2)})
    -- (11.4,\base) -- cycle;
  \draw[npiink!40,line width=.45pt] (.6,\base) -- (11.4,\base);
  \draw[npiink,line width=.9pt]
    plot[domain=-4:4,samples=180]
      ({6+1.35*\x},{\base+2.15*exp(-(\x-(\mu))^2/2)});
}
\node[anchor=west,font=\npiEmphasisFont\bfseries,text=black]
  at (.6,7.56) {Actual class 1};
\node[anchor=west,font=\npiEmphasisFont\bfseries,text=black]
  at (.6,4.46) {Actual class 2};
\node[anchor=east,font=\npiBaseFont,text=npiink!65]
  at (11.4,7.65) {};

% Label correct and erroneous decision masses. Grey mass is excluded.
\node[text=npierror,font=\npiEmphasisFont] (eone) at (3.85,5.48) {$E_1$};
\draw[->,npierror,line width=.6pt] (eone.south east) to[bend right=12] (4.73,5.14);
\node[text=npicorrect,font=\npiEmphasisFont] at (7.85,5.87) {$C_1$};
\node[text=npicorrect,font=\npiEmphasisFont] at (4.20,2.77) {$C_2$};
\node[text=npierror,font=\npiEmphasisFont] (etwo) at (8.48,2.62) {$E_2$};
\draw[->,npierror,line width=.6pt] (etwo.south west) to[bend left=12] (7.27,2.04);
\draw[npiink!60,dashed,line width=.6pt] (5.19,1.49) -- (5.19,7.40);
\draw[npiink!60,dashed,line width=.6pt] (6.81,1.49) -- (6.81,7.40);

% The NP cutoff is schematic; NPI need not bracket it in general.
\draw[npbaseline,dash pattern=on 3pt off 1.4pt on .5pt off 1.4pt,line width=.7pt] (5.61991,1.49) -- (5.61991,7.39);
\node[anchor=south,text=npbaseline,font=\npiBaseFont] at (5.61991,7.39) {Baseline NP cutoff};
% Selective error definitions sit in the open space within each density panel.
\node[text=npiink,font=\npiBaseFont] at (2.70,6.98)
  {Selective type I error};
\node[text=npiink,font=\npiEmphasisFont] at (2.70,6.24)
  {$R_1=
    \dfrac{\textcolor{npierror}{E_1}}
    {\textcolor{npierror}{E_1}+\textcolor{npicorrect}{C_1}}
    \le \alpha_1$};
\node[text=npiink,font=\npiBaseFont] at (9.25,4.07)
  {Selective type II error};
\node[text=npiink,font=\npiEmphasisFont] at (9.25,3.32)
  {$R_2=
    \dfrac{\textcolor{npierror}{E_2}}
    {\textcolor{npierror}{E_2}+\textcolor{npicorrect}{C_2}}
    \le \alpha_2$};

% Shared score axis, followed directly by the NPI area legend.
\draw[->,npiink,line width=.7pt] (.6,1.59) -- (11.62,1.59);
\node[anchor=north,text=npiink] at (5.19,1.46) {$\ell$};
\node[anchor=north,text=npiink] at (6.81,1.46) {$u$};
\node[anchor=east,text=npiink,font=\npiBaseFont] at (11.4,1.31) {score $s(x)$};

% State explicitly that density fills describe the NPI rule.
\node[anchor=west,text=npiink] at (.60,.62) {};
\fill[npicorrect!40] (2.08,.51) rectangle (2.36,.74);
\node[anchor=west,text=npiink] at (2.47,.62) {Correct ($C_i$)};
\fill[npierror!40] (5.45,.51) rectangle (5.73,.74);
\node[anchor=west,text=npiink] at (5.84,.62) {Errors ($E_i$)};
\fill[npimiddle!30] (8.30,.51) rectangle (8.58,.74);
\node[anchor=west,text=npiink] at (8.69,.62) {Indecisions};
\end{tikzpicture}
\endgroup